\documentclass[11pt,a4paper]{article}
\usepackage{amssymb,color}
\usepackage{amsfonts}
\usepackage{amsmath}
\usepackage{amsthm}
\usepackage{graphicx}
\usepackage[dvips]{epsfig}
\usepackage{longtable}
\usepackage{graphics}
\usepackage{color}
\usepackage{enumitem}
\usepackage{subcaption}

\usepackage{graphicx}				

\usepackage{multirow}

\newtheorem{theorem}{Theorem}
\newtheorem{lemma}{Lemma}
\newtheorem{definition}{Definition}

\newtheorem{remark}{Remark}

\newcommand{\trans}{{\mathrm{T}}}
\newcommand{\E}{\mathbb{E}}

\newcommand{\VaR}{\operatorname{V@R}}

\newcommand{\AVaR}{\operatorname{AV@R}}

\DeclareMathOperator*{\esssup}{ess\,sup}
\DeclareMathOperator*{\essinf}{ess\,inf}
\renewcommand{\P}{\mathbb{P}}
\newcommand{\Q}{\mathbb{Q}}
\newcommand{\R}{\mathbb{R}}

\newcommand{\Var}{\operatorname{Var}}

\renewcommand{\phi}{\varphi}
\newcommand{\calF}{\mathcal{F}}

\newcommand{\filF}{\mathbb{F}}

\newcommand{\calP}{\mathcal{P}}

\newcommand{\calS}{\mathcal{S}}

\newcommand{\calQ}{\mathcal{Q}}
\newcommand{\calD}{\mathcal{D}}
\newcommand{\rmd}{\mathrm{d}}
\newcommand{\indic}{\mathbb{I}}

\title{Multiple-prior valuation of cash flows subject to capital requirements}
\author{Hampus Engsner\footnote{hampus.engsner@math.su.se, Stockholm University}, Filip Lindskog\footnote{corresponding author: lindskog@math.su.se, Stockholm University}, Julie Th{\o}gersen\footnote{juliethoegersen@econ.au.dk, Aarhus University}}

\begin{document}
\maketitle
\begin{abstract}
We study market-consistent valuation of liability cash flows motivated by current regulatory frameworks for the insurance industry. Building on the theory on multiple-prior optimal stopping we propose a  valuation functional with sound economic properties that applies to any liability cash flow. 
Whereas a replicable cash flow is assigned the market value of the replicating portfolio, a cash flow that is not fully replicable is assigned a value which is the sum of the market value of a replicating portfolio and a positive margin. 
The margin is a direct consequence of considering a hypothetical transfer of the liability cash flow from an insurance company to an empty corporate entity set up with the sole purpose to manage the liability run-off, subject to repeated capital requirements, and considering the valuation of this entity from the owner's perspective taking model uncertainty into account. 
Aiming for applicability, we consider a detailed insurance application and explain how the optimisation problems over sets of probability measures can be cast as simpler optimisation problems over parameter sets corresponding to parameterised density processes appearing in applications. 
\end{abstract}

\section{Introduction}\label{intro}

We consider the valuation of an aggregate insurance liability cash flow in run-off. The valuation approach is a direct consequence of considering a hypothetical transfer of the liability cash flow from an insurance company to an empty corporate entity set up with the sole purpose to manage the liability run-off. The owner of this entity needs to make sure at any time, in order to continue ownership of the entity, to pay claims and also to provide buffer capital according to the externally imposed  solvency capital requirement (e.g.~by a regulatory framework such as Solvency II). The owner accepts ownership given a suitable initial compensation from the original insurance company wanting to transfer its liabilities. This compensation determine the value of the liability cash flow. However, the owner has the right to, at any time, any surplus exceeding what is required to manage the liability run-off and meet solvency capital requirements. Therefore, the amount of compensation that the owner finds acceptable depends of the owner's view of such surplus.  

The setting and the valuation approach we consider are similar to those considered in Engsner et al.~\cite{Engsner-Lindensjo-Lindskog-20}. An essential difference here is that we acknowledge that an agent who assigns a value to possible future dividends and capital injections from managing a run-off of a liability may consider a valuation functional depending on a set of pricing measures, in the incomplete market setting, rather than a single one. The agent is uncertain about which pricing measure to use and may change view depending on new information. Although this appears to be a modest difference it on the one hand lead to significant mathematical challenges and on the other hand it means that the conservative valuation functional, corresponding to expected discounted values according to the worst pricing measure, can be applied to a wide range of liability cash flows rather than having to be chosen in order to match a specific type of liability cash flow. In order to make this statement clear we may think of cash flows from life insurance. If the agent benefits from survival of policyholders, then a conservative valuation from the agents perspective corresponds to choosing a pricing measure $\Q$ that assigns higher probability to the occurrences of deaths compared to $\P$. However, if the agent instead benefits from deaths of policyholders, then such a $\Q$ no longer corresponds to conservative valuation.  

Insurance liability cash flows may be partly defined in terms of financial asset prices, specific interest rates or inflations indices. For liability cash flows where this is not the case, the cash flows may show significant correlation with market prices. Therefore, any insurance liability valuation methodology must be such not to introduce arbitrage opportunities and must consider replicating portfolios that hedge the financial component of a liability cash flow, whenever that is relevant. Consequently, there is a significant literature on market-consistent insurance valuation covering single-period, multiple-period and continuous-time valuation problems with varying assumptions on the financial market forming the basis for designing replicating portfolios of varying degrees of sophistication. 
We refer (in chronological order) to 
Grosen and J{\o}rgensen \cite{Grosen-Jorgensen-02},  
Malamud et al.~\cite{Malamud-Trubowitz-Wuthrich-08}, 
W\"uthrich et al.~\cite{Wuthrich-Embrechts-Tsanakas-11}, 
M\"ohr \cite{Moehr-11}, 
Tsanakas et al.~\cite{Tsanakas-Wuthrich-Cerny-13}, 
W\"uthrich and Merz \cite{Wuthrich-Merz-13}, 
Pelsser and Stadje \cite{Pelsser-Stadje-14},    
Engsner et al.~\cite{Engsner-Lindholm-Lindskog-17}, 
Delong et al.~\cite{Delong-Dhaene-Barigou-19}, 
Barigou and Dhaene \cite{Barigou-Dhaene-19},  
Engsner et al.~\cite{Engsner-Lindensjo-Lindskog-20}, 
and references therein. 

A common theme in the literature on market-consistent insurance valuation is that the value assigned to a liability cash flow can be expressed as the sum of a market price of a replicating portfolio and a value assigned to the replication error (notice that a substantial replication error is a common feature of insurance liabilities). The liability values in this paper are also of this kind. 
Rebalancing times of a dynamic replicating portfolio means that the replication error has to be reassessed over time and taking this into consideration leads to the notion of time-consistent valuation. Similarly, repeated capital requirements lead to capital costs that are not known at the initial valuation time and taking such costs into account appropriately also require time-consistent valuation. 
Time consistency is a key concept in the literature on dynamic risk measurement. 
We refer (in chronological order) to  
Riedel \cite{Riedel-04}, 
Detlefsen and Scandolo \cite{Detlefsen-Scandolo-05}, 
Rosazza Gianin \cite{RosazzaGianin-06}, 
Cheridito et al.~\cite{Cheridito-Delbaen-Kupper-06},  
Artzner et al.~\cite{Artzner-Delbaen-Eber-Heath-Ku-07}, 
Bion-Nadal \cite{Bion-Nadal-08}, 
Cheridito and Kupper \cite{Cheridito-Kupper-09}, 
Cheridito and Kupper \cite{Cheridito-Kupper-11}, 
and references therein. 

In Artzner et al.~\cite{Artzner-Eisele-Schmidt-20} and Deelstra et al.~\cite{Deelstra-et-al-20} it is argued that diversifiable insurance risk should only be assigned a value corresponding to the $\P$-expectation of such risk since the law of large numbers applies if the insurance company may form arbitrarily large portfolios. In our setting this argument is not valid since the corporate entity to which the insurance company's aggregate liability is transferred is a separate entity (referred to as reference undertaking in Solvency II) that may not be merged with other corporate entities. In that sense the entity to which the liabilities are transferred may be seen as a special purpose vehicle. Although this entity benefits from diversification when capital requirements are computed, it can not diversify the liability further.   

Optimal stopping with multiple priors for agents assessing risk in terms of dynamic convex risk measures is analyzed in Cheridito et al.~\cite{Cheridito-Delbaen-Kupper-06}. Similar problems are analyzed in Engelage \cite{Engelage-11}, where the framework of optimal stopping with multiple priors in \cite{Riedel-09} is extended to so-called dynamic variational preferences. From an applied perspective: whereas all priors/probability measures in a given set of priors are treated as equally likely in the framework in Riedel \cite{Riedel-09}, introducing (dynamic) penalty terms as in Cheridito et al.~\cite{Cheridito-Delbaen-Kupper-06} and Engelage \cite{Engelage-11} means that the optimizing agent may assign different (dynamic) weights to the priors in the optimization problem. Optimal stopping is a key element in our approach to valuation since the owner of the entity managing the run-off of the liability, just as shareholders in general, has limited liability. At any time, taking the value of assets and future liability cash flows into account, if a capital injection is needed to meet capital requirements, the owner may choose between making a capital injection of not. Without such a capital injection, ownership is terminated and the the remaining assets are transferred to policyholders. Therefore, the rational owner determines optimal stopping times.    

Although we assume that the replicating portfolio is chosen to ensure that the valuation of the liability cash flow is market consistent, we do not discuss market consistency in detail since this was treated in detail in Engsner et al.~\cite{Engsner-Lindensjo-Lindskog-20} and the material on market consistency in Engsner et al.~\cite{Engsner-Lindensjo-Lindskog-20} applies without any modification also in the present paper. However, we emphasise that we advocate choosing a replicating portfolio in agreement with what recommended by EIOPA in \cite[Article 38]{Commission-del-reg-15} in Article 38(h) on the Reference Undertaking: "the assets are selected in such a way that they minimise the Solvency Capital Requirement for market risk that the reference undertaking is exposed to". The reference undertaking in Solvency II is similar in spirit to the corporate entity managing the run-off of the liability in the present paper. 

The paper is organised as follows. 
Section \ref{sec:framework} presents the valuation framework. Basic assumptions, notation and terminology are introduced in Subsection \ref{sec:prel}. Subsection \ref{sec:gcr} introduces the agents involved and explains how capital requirements and limited liability are key ingredients in the valuation philosophy that originates from the idea of a hypothetical transfer of an insurance company's liabilities to an empty corporate entity. 
Definitions and results are presented in Subsection \ref{sec:gcr} for very general capital requirements.  
Subsection \ref{sec:cmrmcr} then specialises by considering capital requirements given in terms of conditional monetary risk measures, in line with current regulatory frameworks. 
Section \ref{sec:Q} presents a general construction of a parametric set of priors that cover natural choices for applications and shows that the set of priors satisfies the properties making it suitable for optimal stopping with multiple priors. 
Section \ref{sec:gaussian_example} considers a specific insurance application that illustrates the use of the valuation framework and the results presented.  

\section{The valuation framework}\label{sec:framework}

\subsection{Preliminaries}\label{sec:prel}

We consider time periods $1,\dots,T$, corresponding time points $0,1,\dots,T$, and a filtered probability space $(\Omega,\calF,\filF,\P)$, where $\filF=(\calF_t)_{t=0}^{T}$ with $\{\emptyset, \Omega\}=\calF_0\subseteq \dots \subseteq \calF_{T}=\calF$, and $\P$ denotes the real-world measure. 
For $p\in [1,\infty)$, we write $L^p(\calF_t,\P)$ for the normed linear space of $\calF_t$-measurable random variables $X$ with norm $\E^{\P}[|X|^p]^{1/p}$. We write $L^{\infty}(\calF_t,\P)$ for the normed linear space of $\calF_t$-measurable essentially bounded random variables. 
Equalities and inequalities between random variables should be interpreted in the $\P$-almost sure sense. 
A stopping time is a function $\tau:\Omega\to\{0,1,\dots,T\}\cup \{+\infty\}$ such that $\{\tau=t\}\in\calF_t$ for $t=0,1,\dots,T$. 

For two probability measures $\Q^{(1)},\Q^{(2)}$ equivalent to $\P$ and a stopping time $\tau\leq T$, the probability measure 
$\Q^{(3)}(A):=\E^{\Q^{(1)}}[\Q^{(2)}(A\mid\calF_{\tau})]$, $A\in\calF_T$, is called the pasting of $\Q^{(1)}$ and $\Q^{(2)}$ in $\tau$. 
It is often convenient to express the pasting $\Q^{(3)}$ of $\Q^{(1)},\Q^{(2)}$ in $\tau$ in terms of the density processes $D^{(1)},D^{(2)}$ with respect to $\P$,
\begin{align*}
D^{(1)}_t=\frac{d\Q^{(1)}}{d\P}\Big|_{\calF_t}, \quad D^{(2)}_t=\frac{d\Q^{(2)}}{d\P}\Big|_{\calF_t}.
\end{align*}
The density process $D^{(3)}$ given by 
\begin{align*}
D^{(3)}_t=\indic\{t\leq \tau\}D^{(1)}_t+\indic\{t> \tau\}\frac{D^{(1)}_{\tau}D^{(2)}_t}{D^{(2)}_{\tau}}
\end{align*}
corresponds to the pasting $\Q^{(3)}$ of $\Q^{(1)},\Q^{(2)}$ in $\tau$. Equivalently, we can write
\begin{equation*}
D^{(3)}_t=\prod_{s=1}^t\bigg(\indic\{s\leq \tau\}\frac{D^{(1)}_s}{D^{(1)}_{s-1}}+\indic\{s> \tau\}\frac{D^{(2)}_s}{D^{(2)}_{s-1}}\bigg).
\end{equation*}
A set $\mathcal{Q}$ of probability measures equivalent to $\P$ is called stable under pasting if for any $\Q^{(1)},\Q^{(2)}\in\mathcal{Q}$ and any stopping time $\tau\leq T$, the pasting $\Q^{(3)}$ of $\Q^{(1)},\Q^{(2)}$ in $\tau$ is an element in $\mathcal{Q}$. We call such a set stable under pasting. Such sets are also referred to as m-stable, time consistent or rectangular in the related literature. 

We assume the existence of a financial market containing assets for which $\filF$-adapted price processes $(S^0_t)_{t=0}^T$ and $(S^i_t)_{t=0}^T$, $i=1,\dots,d$, are available. $(S^0_t)_{t=0}^T$ is the price process of a (predictable) locally riskless bond. The price processes correspond to traded assets for which reliable price quotes are available.  
We take the price process of the locally riskless bond as num\'eraire process and in what follows all financial values are discounted by this num\'eraire. This saves us from having to explicitly take interest rates processes into account and makes the mathematical expressions less involved.  
We will also allow for $\filF$-adapted cash flows that depend on insurance events independent of the filtration generated by the traded assets. In particular, we consider an incomplete market setting. 
We assume that the set $\mathcal{P}$ of equivalent martingale measures (for each $\Q\in\mathcal{P}$, $\Q$ is equivalent to $\P$ and the $(S^0_t)_{t=0}^T$-discounted price processes are $\Q$-martingales) is non-empty. By Proposition 6.43 in F\"ollmer and Schied \cite{Foellmer-Schied-16} the set $\mathcal{P}$ is stable under pasting. We will consider a non-empty subset $\mathcal{Q}\subset \mathcal{P}$. We refer to $\Q\in\calQ$ as a market risk neutral probability measure. 
We use the conventions $\sum_{l=k}^{k-1}:=0$ and $\inf\emptyset:=+\infty$ for sums over an empty index set and the infimum of an empty set. We use the notation $(x)^+:=\max(0,x)$.

\subsection{Valuation with general capital requirements}\label{sec:gcr}

We consider an insurance company with an aggregate insurance liability corresponding to a liability cash flow given by the $\filF$-adapted stochastic process $X^o=(X^o_t)_{t=1}^{T}$. Regulation forces the insurance company to comply with externally imposed capital requirements. The requirements put restrictions on the asset portfolio of the insurance company. 
A subset of the assets forms a replicating portfolio with $\filF$-adapted cash flow $X^r=(X^r_t)_{t=1}^T$ intended to, to some extent, offset the liability cash flow. Depending on the degree of replicability of the liability cash flow, the replicating portfolio could be anything from simply a position in the num\'eraire asset to a portfolio that is rebalanced dynamically according to some strategy. $X:=X^{o}-X^{r}$ is the residual liability cash flow.  
We will, in accordance with current solvency regulation 
(M\"ohr \cite{Moehr-11} and prescribed by EIOPA, see \cite[Article 38]{Commission-del-reg-15})
define the value of the liability cash flow $X^o$ by considering a hypothetical transfer of the liability and the replicating portfolio to a separate entity referred to as a reference undertaking. The reference undertaking has initially neither assets nor liabilities and its sole purpose is to manage the run-off of the liability. 
The benefit of ownership is the right to receive certain dividends/surplus, defined below, until either the run-off of the liability cash flow is complete or until letting the reference undertaking default on its obligations to the policyholders. The term default means termination of ownership of the reference undertaking.
The precise details are as follows,

\begin{itemize}
\item 
At time $t=0$: The liabilities corresponding to the cash flow $X^o$, the replicating portfolio corresponding to the cash flow $X^r$ and an amount $R_0$ in the num\'eraire are supposed to be transferred from the insurance company to the reference undertaking, where $R_0$ is the amount making the reference undertaking precisely meet the externally imposed capital requirement. In return, an agent aspiring ownership of the reference undertaking must first pay the original insurance company an amount $C_0$ corresponding to the value of receiving future dividends resulting from managing the run-off of the liability. In case there are several agents aspiring ownership, the one offering the highest amount $C_0$ wins the ownership. 

\emph{In summary: the new owner of the reference undertaking receives compensation $V_0:=R_0-C_0$ from the original insurance company as compensation for accepting to receive the liabilities and replicating portfolio and agreeing to manage the liability run-off.}  
\item 
At time $t=0$: By paying the amount $C_0$ to the original insurance company, the owner receives full ownership of the reference undertaking. However, the cash-flow $X^r$ of the replicating portfolio (possibly defined in terms of a dynamic strategy) cannot be modified by the owner, for instance in order to boost dividend payments in a way that may not be in the interest of policyholders.    
\item 
At time $t=1$: The owner has the option to either default on its obligations to the policyholders or not to default. 

The decision to default means to give up ownership and transfer $R_0$ and the replicating portfolio to the policyholders. The owner neither receives any dividend payment nor incurs any loss upon a decision to default. 

If $T>1$ and given the decision not to default, a new amount $R_1$ in the num\'eraire asset is needed to make the reference undertaking precisely meet the externally imposed capital requirement. If $R_0-R_1-X^{o}_1+X^{r}_1\geq 0$, then the positive surplus $R_0-R_1-X^{o}_1+X^{r}_1\geq 0$ is paid to the owner and $X^{o}_1$, which the policyholders are entitled to, is paid to the policyholders. If $R_0-R_1-X^{o}_1+X^{r}_1<0$, then the owner faces a deficit that must be offset by injecting $-R_0+R_1+X^{o}_1-X^{r}_1>0$. Also in this case $X^{o}_1$ is paid to the policyholders. 

If $T=1$, then the above description of cash flows to policyholders and owner applies upon setting $R_1=0$.

\item At time $t\in \{2,\dots,T\}$: If the owner has not defaulted on its obligations, then the situation is completely analogous to that at time $t=1$ described above.
\end{itemize}

From the above follows that the owner of the reference undertaking has to decide on a decision rule defining under which circumstances default occurs. The default time is a stopping time $\tau \in \calS_{1,T+1}$, where $\calS_{t,T+1}$ denotes the set of $\filF$ stopping times taking values in $\{t,\dots,T+1\}$. The event $\{\tau=T+1\}$ is to be interpreted as a complete liability run-off without default at any time. Formally,
$\calS_{t,T+1}:=\{\tau : \tau \text{ is a stopping time with } \tau\geq t\} \wedge (T+1)$. 

The cumulative cash flow to the owner can be written as
\begin{align}
\sum_{t=1}^{\tau-1}(R_{t-1} - R_t -X_t), \quad X_t:=X^{o}_t-X^{r}_t. \label{C-cashflow}
\end{align}
For ease of notation, define the payoff process $(H_t)_{t=1}^T$ by 
\begin{align}
H_1:=0, \quad H_t:=\sum_{s=1}^{t-1}(R_{s-1} - R_s -X_s) \quad\text{for } t>1.\label{H-process}
\end{align}
Note that this payoff process is predictable. The conservative value of the cash flow \eqref{C-cashflow} is 
\begin{align}
\inf_{\Q \in \calQ}\E_0^{\Q}[H_\tau]. \label{C-value-tau}
\end{align}
We assume that the owner of the reference undertaking chooses a default time $\tau$ maximizing the value \eqref{C-value-tau}. Consequently, the value at time $0$ of the reference undertaking is 
\begin{align}\label{C0-expression}
\sup_{\tau \in \mathcal{S}_{1,T+1}}\inf_{\Q \in \calQ}\E_0^{\Q}[H_\tau].
\end{align}
For $t\in\{1,\dots,T\}$, the value of the reference undertaking at time $t$, given no default at times $\leq t$, is given by the completely analogous expression upon replacing $\sup$ and $\inf$ in \eqref{C0-expression} by the essential supremum $\esssup$ and essential infimum $\essinf$ (see Appendix A.5 in F\"ollmer and Schied \cite{Foellmer-Schied-16} for details) and conditioning on $\calF_t$ rather than $\calF_0$. 
Notice that since no cash flows occur at times $>T$, the value of the reference undertaking is zero at time $T$.
The value of the reference undertaking can thus be identified as the value of an American type derivative. Details on arbitrage-free pricing of American derivatives can be found in Section 6.3 in \cite{Foellmer-Schied-16}. 

Since we are considering sets $\calQ$ of probability measures we need the cash flows to be suitably integrable with respect to all $\Q\in\calQ$. The following notion of uniform integrability, from Riedel \cite{Riedel-09}, will be used.  
The process $(H_t)_{t=1}^T$ in \eqref{H-process} is bounded by a $\calQ$-uniformly integrable random variable in the sense that there exists $Z\geq 0$ such that  
\begin{align}\label{eq:H-uniformly-bounded}
\sup_{t\in\{1,\dots,T\}} |H_t|\leq Z
\text{ and }
\lim_{K\to\infty}\sup_{\Q\in\calQ}\E^{\Q}[Z\indic_{\{Z\geq K\}}]=0.
\end{align}   

We now define the value of the reference undertaking, corresponding to what an external party would pay to become owner of the entity managing the run-off of the liability, and also the value of the residual liability. The sum of the latter and the market price of the replicating portfolio is the value of the original liability to policyholders and therefore is a theoretical aggregate premium.  

\begin{definition}\label{Ct-definition}
Let $\calQ$ be a set of market risk neutral probability measures.
Consider sequences $(X_t)_{t=1}^T$ and $(R_t)_{t=0}^T$ with $X_t\in L^{1}(\calF_t,\Q)$ for $t\in \{1,\dots,T\}$ for every $\Q\in\calQ$, $R_T=0$ and $R_t\in L^{1}(\calF_t,\Q)$ for $t\in \{0,\dots,T-1\}$ for every $\Q\in\calQ$. Define 
$C_T :=0$ and, for $t\in \{0,\dots,T-1\}$, 
\begin{align}
C_t := \esssup_{\tau \in \mathcal{S}_{t+1,T+1}}\essinf_{\Q \in \calQ}\E_t^{\Q}\bigg[\sum_{s=t+1}^{\tau-1}(R_{s-1}-R_s-X_s) \bigg].  \label{Ct-expression} 
\end{align}
$C_t$ is \emph{the value of the reference undertaking at time $t$} given no default at times $\leq t$. 
$V_t:=R_t-C_t$ is \emph{the value of the residual liability at time $t$} given no default at times $\leq t$. 
\end{definition}

Notice that
\begin{align*}
V_t&:=R_t-C_t\\ 
&= R_t-\esssup_{\tau\in\mathcal{S}_{t+1,T+1}}\essinf_{\Q \in \calQ}\E_t^{\Q}\bigg[\sum_{s=t+1}^{\tau-1}(R_{s-1}-R_s-X_s)\bigg] \\
&=\essinf_{\tau\in\mathcal{S}_{t+1,T+1}}\esssup_{\Q \in \calQ}\E_t^{\Q}\bigg[R_t - \sum_{s=t+1}^{\tau-1}(R_{s-1}-R_s-X_s)\bigg] \\
&=\essinf_{\tau \in \mathcal{S}_{t+1,T+1}}\esssup_{\Q \in \calQ}\E_t^{\Q}\bigg[\sum_{s=t+1}^{\tau-1}X_s+R_{\tau-1}\bigg] \\
&\leq \esssup_{\Q \in \calQ}\E_t^{\Q}\bigg[\sum_{s=t+1}^{T}X_s\bigg]=:\overline{V}_t. 
\end{align*}
The general upper bound 
\begin{align}\label{eq:V0generalupperbound}
\overline{V}_0:=\sup_{\Q \in \calQ}\E_0^{\Q}\bigg[\sum_{s=1}^{T}X_s\bigg]\geq V_0
\end{align}
does neither depend on the filtration nor on the capital requirements, and is typically much easier to compute than $V_0$. Therefore, this upper bound provides a useful conservative estimate of $V_0$. This statement is illustrated in the numerical example in Section \ref{sec:gaussian_example}. 
Notice that in general
\begin{align}
V_t&=\essinf_{\tau \in \mathcal{S}_{t+1,T+1}}\esssup_{\Q \in \calQ}\E_t^{\Q}\bigg[\sum_{s=t+1}^{\tau-1}X_s+R_{\tau-1}\bigg] \nonumber \\
&\geq \esssup_{\Q \in \calQ}\essinf_{\tau \in \mathcal{S}_{t+1,T+1}}\E_t^{\Q}\bigg[\sum_{s=t+1}^{\tau-1}X_s+R_{\tau-1}\bigg] 
=: \underline{V}_t. \label{eq:Vtgenerallowerbound}
\end{align}
In particular, the general lower bound 
\begin{align}\label{eq:V0generallowerbound}
\underline{V}_0:=\sup_{\Q \in \calQ}\inf_{\tau \in \mathcal{S}_{1,T+1}}\E_0^{\Q}\bigg[\sum_{s=1}^{\tau-1}X_s+R_{\tau-1}\bigg]
\leq V_0
\end{align}
may be attractive since it is typically easier to compute than $V_0$, see Section \ref{sec:gaussian_example} for an illustration.  
Computing $\underline{V}_0$ means solving a standard optimal stopping problem for each $\Q\in\calQ$ followed by finding the maximum of the obtained values $V_0^{\Q}$. 

Notice that the value $L_0$ of the original liability cash flow $X^{o}$ follows directly from the procedure for transferring the liabilities and replicating portfolio to an external party (the new owner of the reference undertaking) accepting the transfer: $L_0$ equals the sum of the market value of the replicating portfolio and the value $V_0$ of the residual liability:
$$
L_0=\E_0^{\Q}\bigg[\sum_{s=1}^{T}X^{r}_s\bigg]+V_0,
$$ 
where $\Q$ is any market risk neutral probability measure making the expectation equal the market value of the replicating portfolio. For details on the market consistency of the value $L_0$ we refer to the material on market consistency in Engsner et al.~\cite{Engsner-Lindensjo-Lindskog-20}. 

We intend to build on the theory of multiple prior optimal stopping in Riedel \cite{Riedel-09} where four assumptions on a set $\calQ$ of probability measures are imposed in order for key results to hold. These assumptions are $\calQ$-uniform integrability together with properties (i)-(iii) of the following definition. 

\begin{definition}\label{assumption}
A set $\calQ$ of probability measures is suitable for multiple prior optimal stopping if the following properties hold. 
(i) Each $\Q\in\calQ$ is equivalent to $\P$; 
(ii) $\calQ$ is stable under pasting; 
(iii) For each $t\in \{0,\dots,T\}$, 
$$
\bigg\{\frac{d\Q}{d\P}\bigg|_{\calF_t} : \Q\in\mathcal{Q}\bigg\}
$$
is weakly compact in $L^1(\calF_T,\P)$.
\end{definition}

\begin{remark}
If $\calQ$ satisfies the properties (i)-(iii) in Definition \ref{assumption}, then it follows from Theorem 2 in Riedel \cite{Riedel-09} that the lower bound $\underline{V}_0$ in \eqref{eq:V0generallowerbound} equals $V_0$. This holds since for such $\calQ$ the inequality in \eqref{eq:Vtgenerallowerbound} is in fact a minimax identity. 
Notice also that for an arbitrary $\Q$ equivalent to $\P$, $\{\Q\}$ satisfies properties (i)-(iii) in Definition \ref{assumption}. 
\end{remark}

As a basis for applying the theory to be presented, we will later in Section \ref{sec:Q} explicitly construct a useful set $\calQ$ satisfying the properties in Definition \ref{assumption} and present a detailed numerical example in Section \ref{sec:gaussian_example}. 

We are now ready to state a key result which shows that $(C_t,V_t)$ defined in terms of a multiple prior optimal stopping problem may equivalently be defined as the solution to a backward recursion.    

\begin{theorem}\label{Vt-def-thm} 
Let $\calQ$ be a set of probability measures satisfying properties (i)-(iii) of Definition \ref{assumption}. 
Consider sequences $(X^{o}_t)_{t=1}^T$, $(X^{r}_t)_{t=1}^T$, $(R_t)_{t=0}^T$ with $X^{o}_t,X^{r}_t\in L^{1}(\calF_t,\Q)$ for $t\in \{1,\dots,T\}$ for every $\Q\in\mathcal{Q}$, $R_T=0$ and $R_t\in L^{1}(\calF_t,\Q)$ for $t\in \{0,\dots,T-1\}$ for every $\Q\in\calQ$. Set $X_t:=X^{o}_t-X^{r}_t$ and assume that $(H_t)_{t=1}^T$ in \eqref{H-process} is bounded by a $\calQ$-uniformly integrable random variable. Then  

(i) 
If the sequences $(C_t)_{t=0}^T$ and $(V_t)_{t=0}^T$ are given by Definition \ref{Ct-definition}, then 
\begin{align} 
C_t &= \essinf_{\Q \in \calQ} \E_t^{\Q}[(R_t-X_{t+1}-V_{t+1})^+], \quad C_T=0, \label{Ct-expression2}\\
V_t &= \esssup_{\Q \in \calQ} \Big(R_t-\E_t^{\Q}[(R_t-X_{t+1}-V_{t+1})^+]\Big), \quad V_T=0. \label{Vt-expression2}
\end{align}

(ii) 
The stopping times $(\tau^*_t)_{t=0}^{T-1}$ given by 
\begin{align*} 
\tau^*_{t}=  \inf \{s \in \{t+1,\dots,T\} : R_{s-1}-X_{s}-V_{s}<0\} \wedge (T+1)
\end{align*}
are optimal in \eqref{Ct-expression}. 

(iii) 
If the sequences $(C_t)_{t=0}^T$ and $(V_t)_{t=0}^T$ are given by \eqref{Ct-expression2} and \eqref{Vt-expression2}, then, for $t\in\{0,\dots,T-1\}$, $C_t$ is given by \eqref{Ct-expression} and $V_t=R_t-C_t$. 
\end{theorem}

\begin{remark}
Stability under pasting of $\calQ$ is a necessary requirement in Theorem \ref{Vt-def-thm}. However, we show later in Theorem \ref{thm:general_prior_example} that instead of the weak compactness property (iii) in Definition \ref{assumption}, which is assumed in Theorem \ref{Vt-def-thm}, it is sufficient to verify weak relative compactness together with some natural additional properties. Notice that a bounded and uniformly integrable subset of $L^1(\calF_t,\P)$ is weakly relatively compact in $L^1(\calF_t,\P)$ (Theorem A.70 in \cite{Foellmer-Schied-16}). Without weak compactness we can however not guarantee that there exists a $\Q^*\in\calQ$ which solves the optimization problems \eqref{Ct-expression2} and \eqref{Vt-expression2}. 
\end{remark}

\begin{proof}[Proof of Theorem \ref{Vt-def-thm}]
We will first consider the problem
\begin{align}
\esssup_{\tau \in \calS_{t,T+1}}\essinf_{\Q \in \calQ}\E_t^\Q[H_\tau]. \label{eq:modified_stopping_problem}
\end{align}
We define the multiple prior Snell envelope of $H$ with respect to $\calQ$ as in \cite{Riedel-09} by 
\begin{align}\label{eq:multi-prior-snell1}
U^\calQ_{T+1} := H_{T+1},\quad
U^\calQ_t :=\max\big\{H_t, \essinf_{\Q \in \calQ}\E_t^\Q[U^\calQ_{t+1}]\big\} \quad \text{for } t\leq T.
\end{align}
We know from Theorem 1 in \cite{Riedel-09} that 
\begin{align}\label{eq:multi-prior-snell2}
U^\calQ_{t} = \esssup_{\tau \in \calS_{t,T+1}}\essinf_{\Q \in \calQ}\E_t^\Q[H_\tau]
\end{align}
and that $\tau^*_t := \inf\{s \geq t  : U^{\calQ}_s = H_s\}$ is an optimal stopping time that solves \eqref{eq:modified_stopping_problem}. Define $\widetilde{U}^{\calQ}$ by 
\begin{align*}
\widetilde{U}^{\calQ}_t := \esssup_{\tau \in \calS_{t+1,T+1}}\essinf_{\Q \in \calQ}\E_t^\Q[H_\tau].
\end{align*}
We claim that the relation $\widetilde{U}^{\calQ}_t = \essinf_{\Q \in \calQ}\E_t^\Q[U^\calQ_{t+1}]$ holds.  
Indeed, from \eqref{eq:multi-prior-snell2},
\begin{align*}
U^\calQ_{t} = \max\big\{H_{t},\esssup_{\tau \in \calS_{t+1,T+1}}\essinf_{\Q \in \calQ}\E_t^\Q[H_\tau]\big\}
= \max\big\{H_{t},\widetilde{U}^{\calQ}_t\big\}.
\end{align*}
Therefore, from \eqref{eq:multi-prior-snell1}, we have the relation
\begin{align*}
\max\big\{H_t, \essinf_{\Q \in \calQ}\E_t^\Q[U^\calQ_{t+1}]\big\}=\max\big\{H_{t},\widetilde{U}^{\calQ}_t\big\}.
\end{align*} 
Since this holds for arbitrary adapted $H$, the claim is proved and gives
\begin{align*}
C_t &= \widetilde{U}^{\calQ}_t -H_{t+1} =\essinf_{\Q \in \calQ}\E_t^\Q[U^\calQ_{t+1}] -H_{t+1} \\
&= \essinf_{\Q \in \calQ}\E_t^\Q[ \max\{H_{t+1},\essinf_{\Q \in \calQ}\E_{t+1}^\Q[U^\calQ_{t+2}]\}-H_{t+1}]\\
&= \essinf_{\Q \in \calQ}\E_t^\Q[ \max\{0,\essinf_{\Q \in \calQ}\E_{t+1}^\Q[U^\calQ_{t+2}]-H_{t+1}\}]\\
&= \essinf_{\Q \in \calQ}\E_t^\Q[ \max\{0,C_{t+1} + H_{t+2} - H_{t+1}\}]\\
&= \essinf_{\Q \in \calQ}\E_t^\Q[ (R_t - X_{t+1}- V_{t+1})^+].
\end{align*}
Hence, we have shown \eqref{Ct-expression2} from which \eqref{Vt-expression2} is an immediate consequence. This concludes the proof of statement (i). 
\end{proof}

\subsection{Valuation with capital requirements by conditional monetary risk measures}\label{sec:cmrmcr}

We now consider the valuation problem in the setting where the capital requirements are given in terms of 
conditional monetary risk measures.

\begin{definition}\label{def:dynrisk}
For $p\in [0,\infty]$ and $t\in\{0,\dots,T-1\}$, a conditional monetary risk measure is a mapping $\rho_t:L^p(\calF_{t+1},\P)\to L^p(\calF_t,\P)$ satisfying
\begin{align}
& \textrm{if } \lambda\in L^p(\calF_t,\P) \textrm{ and } Y\in L^p(\calF_{t+1},\P), \textrm{ then } 
\rho_t(Y+\lambda)=\rho_t(Y)-\lambda,  \label{eq:ci_r}\\
& \textrm{if } Y,\widetilde{Y}\in L^p(\calF_{t+1},\P) \textrm{ and } Y\leq \widetilde{Y}, \textrm{ then } 
\rho_t(Y)\geq \rho_t(\widetilde{Y}),  \label{eq:mo_r}\\
& \rho_t(0)=0. \label{eq:no_r}
\end{align}
A sequence $(\rho_t)_{t=0}^{T-1}$ of conditional monetary risk measures is called a dynamic monetary risk measure.
\end{definition}

The natural conditional monetary risk measures corresponding to current regulatory frameworks are defined in terms of conditional quantile functions. For integer $t\geq 0$, $x\in\R$, $u\in (0,1)$ and $\calF_{t+1}$-measurable $Z$, let 
\begin{align*}
F_{t,-Z}(x)&:=\P_t(-Z\leq x),\\
F_{t,-Z}^{-1}(1-u)&:=\essinf\{m\in L^0(\calF_t,\P):F_{t,-Z}(m)\geq 1-u\}
\end{align*}
and define the conditional versions of value-at-risk and average value-at-risk as 
\begin{align*}
\VaR_{t,u}(Z):=F_{t,-Z}^{-1}(1-u), \quad \AVaR_{t,u}(Z):=\frac{1}{u}\int_0^u\VaR_{t,v}(Z)dv.
\end{align*}
Both $\VaR_{t,u}$ and $\AVaR_{t,u}$ are conditional monetary risk measures in the sense of Definition \ref{def:dynrisk} for $p\geq 1$. 
Given conditional monetary risk measures $\rho_t:L^1(\calF_{t+1},\P)\to L^1(\calF_t,\P)$ we consider here 
\begin{align}\label{eq:Rtbyrhot}
R_t := \rho_t(-X_{t+1}-V_{t+1}), \quad R_T:=0.
\end{align}
Notice that if $R_{t+1}$ is given and $C_{t+1}$ is given by Definition \ref{Ct-definition}, then also $V_{t+1}:=R_{t+1}-C_{t+1}$ is given and therefore $R_t$ is well defined by setting $R_t := \rho_t(-X_{t+1}+V_{t+1})$. Moreover, we may write 
\begin{align}\label{eq:Vtbyrhot}
V_t := \phi_t(X_{t+1}+V_{t+1}), \quad V_T:=0,
\end{align}
where 
\begin{align*}
\phi_t(Y) := \rho_t(-Y)-\essinf_{\Q \in \calQ} \E_t^{\Q}[(\rho_t(-Y)-Y)^+].
\end{align*}

\begin{theorem}\label{phitproperties}
Let $\rho_t:L^1(\calF_{t+1},\P)\to L^1(\calF_t,\P)$ be a conditional monetary risk measure in the sense of Definition \ref{def:dynrisk} and let $\phi_t:L^1(\calF_{t+1},\P)\to L^1(\calF_t,\P)$ be given by \eqref{eq:Vtbyrhot}. Then 
\begin{align}
& \textrm{if } \lambda\in L^p(\calF_t,\P) \textrm{ and } Y\in L^p(\calF_{t+1},\P), \textrm{ then } 
\phi_t(Y+\lambda)=\phi_t(Y)+\lambda,  \label{eq:ci_v}\\
& \textrm{if } Y,\widetilde{Y}\in L^p(\calF_{t+1},\P) \textrm{ and } Y\leq \widetilde{Y}, \textrm{ then } 
\phi_t(Y)\leq \phi_t(\widetilde{Y}),  \label{eq:mo_v}\\
& \phi_t(0)=0. \label{eq:no_v}
\end{align}
Moreover, if $\widetilde{\rho}_t:L^1(\calF_{t+1},\P)\to L^1(\calF_t,\P)$ is a conditional monetary risk measure in the sense of Definition \ref{def:dynrisk} such that $\rho_t\leq \widetilde{\rho}_t$, then 
\begin{align*}
\widetilde{\phi}_t(Y) := \widetilde{\rho}_t(-Y)-\essinf_{\Q \in \calQ} \E_t^{\Q}[(\widetilde{\rho}_t(-Y)-Y)^+]
\end{align*}
satisfies $\phi_t\leq \widetilde{\phi}_t$.
\end{theorem}
\begin{proof}[Proof of Theorem \ref{phitproperties}]
The properties \eqref{eq:ci_v}, \eqref{eq:mo_v} and \eqref{eq:no_v} follow immediately by arguments similar to those in the proof of Proposition 1 in Engsner et al.~\cite{Engsner-Lindholm-Lindskog-17}. The final property follows immediately from the inequality 
$$
(\widetilde{R}-Y)^+\leq (R-Y)^++\widetilde{R}-R \quad \text{for } \widetilde{R}\geq R.
$$
\end{proof}
Theorem \ref{phitproperties} has consequences that should be seen as necessary requirements of any sound valuation method. If $X_1+\dots+X_T=c$ for some constant $c$, then the corresponding value $V_0=c$. If we consider two residual liability cash flows $(X_t)_{t=1}^T$ and $(\widetilde{X}_t)_{t=1}^T$ such that $X_t\leq \widetilde{X}_t$ for every $t$, then the corresponding values satisfy $V_0\leq \widetilde{V}_0$. Similarly, if the sequence of conditional monetary risk measures $(\rho_t)_{t=0}^{T-1}$
are replaced by a more prudent choice $(\widetilde{\rho}_t)_{t=0}^{T-1}$ such that $\rho_t\leq \widetilde{\rho}_t$ for every $t$, then the corresponding values satisfy $V_0\leq \widetilde{V}_0$. 

Let $(V^S_t)_{t=0}^T$ be given by $V^S_t=\sum_{u=1}^t X_u + V_t$, where $V_T=0$, 
$$
V_t=R_t-\essinf_{\Q\in\mathcal{Q}}\E^{\Q}_t[(R_t-X_{t+1}-V_{t+1})_+], \quad R_t=\rho_t(-X_{t+1}-V_{t+1}),
$$
where $\rho_t$ is a suitable conditional monetary risk measure such as $\VaR_t$ or $\AVaR_t$.
It is reasonable to require that $V^S$ is a $\P$-supermartingale which is equivalent to $V_t\geq \E^{\P}_t[X_{t+1}+V_{t+1}]$ which implies $V_t\geq \E^{\P}_t[X_{t+1}+\dots+X_{T}]$. In particular, the $\P$-supermartingale property guarantees the existence of a nonnegative ``risk margin" 
$V_t-\E^{\P}_t[X_{t+1}+\dots+X_{T}]\geq 0$.

\begin{theorem}\label{thm:Vsupermartingale}
Let $X_t\in L^1(\calF_t,\P)$ for $t=1,\dots,T$. Let $L^1(\calF_{t+1},\P)\ni Y_{t+1}\mapsto \rho_t(-Y_{t+1})\in L^1(\calF_{t},\P)$, for $t=0,\dots,T-1$, be a conditional monetary risk measure such that 
\begin{align}\label{eq:crm_cond}
\E^{\P}_t[\rho_t(-Y_{t+1})-Y_{t+1}]>0 \quad\text{or}\quad \P_t(\rho_t(-Y_{t+1})-Y_{t+1}=0)=1.
\end{align} 
Then there exists a set $\calQ$ of probability measures such that $(V^S_t)_{t=0}^{T}$ is a $\P$-supermartingale.  
\end{theorem}

\begin{proof}[Proof of Theorem \ref{thm:Vsupermartingale}]
Notice that the supermartingale requirement is equivalent to 
\begin{align}\label{eq:ineq}
\essinf_{\Q\in\mathcal{Q}}\E^{\Q}_t[(R_t-X_{t+1}-V_{t+1})_+]\leq \E^{\P}_t[R_t-X_{t+1}-V_{t+1}]
\end{align}
It is sufficient to find some $\Q$ such that the statement holds for $\calQ=\{\Q\}$. 
We construct this $\Q$ by defining a suitable $\P$-martingale $(D_{t})_{t=0}^T$ corresponding to the change of measure from $\P$ to $\Q$. 

Let $W_{t+1}:=R_t-X_{t+1}-V_{t+1}$, let $G_t(x):=\P_t(W_{t+1}\leq x)$ denote the $\calF_{t}$-conditional distribution function of $W_{t+1}$, and let $p_t:=G_t(0)$. Let $(D_{t})_{t=0}^T$, with $D_0=1$, be a $\P$-martingale satisfying 
\begin{align*}
\frac{D_{t+1}}{D_t}=
\left\{\begin{array}{ll}
1 & \text{if } p_t\in\{0,1\},\\
\exp\big(\lambda_t\Phi^{-1}(U_{t+1})-\lambda_t^2/2\big) & \text{if } p_t\in (0,1),
\end{array}\right.
\end{align*}
where $U_{t+1}$ is independent of $\calF_t$ and uniformly distributed on $(0,1)$ and, conditional of $\calF_t$, $U_{t+1}$ and $W_{t+1}$ are countermonotone. Let $\lambda_t$ be some $\calF_t$-measurable random variable satisfying 
\begin{align*}
\exp\big(\lambda_t \Phi^{-1}(1-p_t)-\lambda_t^2/2\big)\E^{\P}_t\big[W_{t+1}^+\big]\leq \E^{\P}_t\big[W_{t+1}\big]
\quad \text{on } \{p_t \in (0,1)\}.
\end{align*}
By construction, 
\begin{align*}
\E^{\P}_t\bigg[\frac{D_{t+1}}{D_t}W_{t+1}^+\bigg]=\E^{\P}_t\big[W_{t+1}\big] 
\quad \text{on } \{p_t \in \{0,1\}\}.
\end{align*}
Moreover, on $\{p_t\in (0,1)\}$, 
\begin{align*}
\E^{\P}_t\bigg[\frac{D_{t+1}}{D_t}W_{t+1}^+\bigg]
&=\int_0^{1-p_t}\exp\big(\lambda_t \Phi^{-1}(u)-\lambda_t^2/2\big)G_{t}^{-1}(1-u)du\\
&\leq \exp\big(\lambda_t \Phi^{-1}(1-p_t)-\lambda_t^2/2\big)\int_0^{1-p_t}G_{t}^{-1}(1-u)du\\
&=\exp\big(\lambda_t \Phi^{-1}(1-p_t)-\lambda_t^2/2\big)\E^{\P}_t\big[W_{t+1}^+\big]\\
&\leq \E^{\P}_t\big[W_{t+1}\big].  
\end{align*}
\end{proof}

Property \eqref{eq:crm_cond} in Theorem \ref{thm:Vsupermartingale} is satisfied by $\AVaR_{t,u}$ which is an example of so-called strictly expectation bounded risk measures, see Definition 5 and Example 3 in Rockafellar et al.~\cite{Rockafellar-et-al-06}. 

The following lemma is useful for constructing a bounding $\calQ$-uniformly integrable random variable. 

\begin{lemma}\label{lem:unif_int}
For any $\calQ$-uniformly integrable $Z \geq 0$, $\esssup_{\Q\in\calQ}\E_t^{\Q}[Z]$ is a $\calQ$-uniformly integrable random variable.
\end{lemma}
\begin{proof}[Proof of Lemma \ref{lem:unif_int}]
We need to show that
$$
\lim_{K\to\infty}\sup_{\Q\in\calQ}\E^{\Q}\Big[\esssup_{\Q\in\calQ}\E_t^{\Q}[Z]\indic_{\{\esssup_{\Q\in\calQ}\E_t^{\Q}[Z] \geq K\}}\Big]=0.
$$ 
If we set $X= Z\indic_{\{\esssup_{\Q\in\calQ}\E_t^{\Q}[Z] \geq K\}}$, then $X$ is $\calQ$-uniformly integrable since it is of the form  $Z\indic_A$. Hence by the law of iterated expectations for $\calQ$-uniformly integrable random variables (Lemma 1 in \cite{Riedel-09}), 
\begin{align*}
\sup_{\Q\in\calQ}\E^{\Q}\Big[\esssup_{\Q\in\calQ}\E_t^{\Q}[Z]\indic_{\{\esssup_{\Q\in\calQ}\E_t^{\Q}[Z] \geq K\}}\Big] 
&= \sup_{\Q\in\calQ}\E^{\Q}\Big[\esssup_{\Q\in\calQ}\E_t^{\Q}[X]\Big] \\
&= \sup_{\Q\in\calQ}\E^{\Q}[X].
\end{align*}
Notice that $\sup_{\Q \in \calQ}\E^{\Q}[Z]<\infty$ since for any $r>0$, 
$$
\sup_{\Q \in \calQ}\E^{\Q}[Z]\leq r+\sup_{\Q \in \calQ}\E^{\Q}[|Z|\indic_{|Z|>r}]
$$
and, due to $\calQ$-uniformly integrability of $Z$, we may choose $r$ to make the second term on the right-hand side sufficiently small.   
Since
$$
\sup_{\Q\in\calQ}\E^{\Q}\Big[\esssup_{\Q\in\calQ}\E_t^{\Q}[Z]\Big] =  \sup_{\Q \in \calQ}\E^{\Q}[Z]<\infty,
$$
the events $A_n = \{\esssup_{\Q\in\calQ}\E_t^{\Q}[Z] \geq n\}$ satisfy $\sup_{\Q \in \calQ} \Q(A_n) \to 0$ as $n \to \infty$. For any $\Q\in\calQ$ and $r_n>0$, 
\begin{align}\label{eq:Anrn}
\sup_{\Q\in\calQ}\E^{\Q}\Big[Z\indic_{A_n}\Big]\leq r_n\sup_{\Q\in\calQ}\Q(A_n)+\sup_{\Q\in\calQ}\E^{\Q}\Big[|Z|\indic_{\{|Z|>r_n\}}\Big].
\end{align}
Consider a sequence $(r_n)_{n=1}^{\infty}$ such that $r_n\to\infty$ and $r_n\sup_{\Q \in \calQ} \Q(A_n) \to 0$ as $n \to \infty$.
Applying this sequence to \eqref{eq:Anrn}, taking the supremum over $\calQ$ and letting $n\to\infty$ proves the statement of the lemma.
\end{proof}

The following theorem says that if the conditional monetary risk measures $\rho_t$ defining $R_t$ in \eqref{eq:Rtbyrhot} satisfy natural and verifiable bounds, then statements in Theorem \ref{Vt-def-thm} hold also in this setting. 

\begin{theorem}\label{thm:Vt-def-thm-specialised}
Let $\mathcal{Q}$ be a set of probability measures satisfying properties (i)-(iii) of Definition \ref{assumption}. 
Consider sequences $(X^{o}_t)_{t=1}^T$, $(X^{r}_t)_{t=1}^T$,  with $X^{o}_t,X^{r}_t\in L^{1}(\calF_t,\Q)$ for $t\in \{1,\dots,T\}$ for every $\Q\in\mathcal{Q}$, $R_T=0$ and $R_t\in L^{1}(\calF_t,\Q)$ for $t\in \{0,\dots,T-1\}$ for every $\Q\in\mathcal{Q}$. 
Let $X_t:=X^{o}_t-X^{r}_t$ and let $(R_t)_{t=0}^T$ be defined by \eqref{eq:Rtbyrhot}. Assume that $\sum_{t=1}^T|X_t|$ is $\calQ$-uniformly integrable. If the conditional monetary risk measures $\rho_t$ in \eqref{eq:Rtbyrhot} satisfy either 
\begin{align}\label{eq:risk_measure_bound}
|\rho_t(Z)|\leq K_{\rho}\esssup_{\Q\in\calQ}\E^{\Q}_t[|Z|]
\text{ for some } K_{\rho}\in (1,\infty)
\end{align}
or 
\begin{align}\label{eq:risk_measure_bound2}
\P \in \calQ \text{ and } 
|\rho_t(Z)|\leq K_{\rho}\E^{\P}_t[|Z|]
\text{ for some } K_{\rho}\in (1,\infty),
\end{align}
then $(H_t)_{t=0}^T$ defined in \eqref{H-process} satisfies that $\esssup_{t=1, \dots T} H_t$ is bounded by a $\calQ$-uniformly integrable random variable. In particular, the statements in Theorem \ref{Vt-def-thm} hold.  
\end{theorem}

\begin{proof}[Proof of Theorem \ref{thm:Vt-def-thm-specialised}]
Set 
\begin{align*}
S_T:=0, \quad S_t:=\esssup_{\Q\in\calQ}\E^{\Q}_t\bigg[\sum_{u=t+1}^T |X_u|\bigg] \text{ for } t=0,1,\dots,T-1.
\end{align*}
By Lemma \ref{lem:unif_int} all variables $S_t$ are $\calQ$-uniformly integrable. 
We will show by induction that, for all $t$, there exist constants $K_{V,t},K_{R,t}\in (1,\infty)$ such that   
\begin{align}\label{eq:induc_statement}
|V_{t}| \leq K_{V,t}S_{t}, \quad 
|R_{t}| \leq  K_{R,t}S_{t} 
\end{align}
from which the statement of the theorem follows. Note that \eqref{eq:induc_statement} trivially holds for $t=T$. In order to show the induction step, assume that \eqref{eq:induc_statement} holds with $t$ replaced by $t+1$. 
If \eqref{eq:risk_measure_bound2} holds, then 
\begin{align*}
|R_{t}| &= |\rho_t(-X_{t+1}-V_{t+1})| \\
&\leq K_{\rho} \E^\P_t[|X_{t+1}|+|V_{t+1}|]\\
&\leq K_{\rho} \E^\P_t[|X_{t+1}| +K^V_{t+1} S_{t+1} ] \\
&\leq K_{\rho} K_{V,t+1} \esssup_{\Q \in \calQ} \E^\Q_t[|X_{t+1}|+S_{t+1}] \\
&= K_{\rho} K_{V,t+1}S_t,
\end{align*}
where the law of iterated expectations for $\calQ$-uniformly integrable random variables (Lemma 1 in \cite{Riedel-09}) was used in the last step. 
If \eqref{eq:risk_measure_bound} holds, then similarly 
\begin{align*}
|R_{t}| &= |\rho_t(-X_{t+1}-V_{t+1})| \\
&\leq K_{\rho} \E^\Q_t[|X_{t+1}|+|V_{t+1}|]\\
&\leq K_{\rho} \E^\Q_t[|X_{t+1}| +K_{V,t+1} S_{t+1} ] \\
&\leq K_{\rho} K_{V,t+1} \esssup_{\Q \in \calQ} \E^\Q_t[|X_{t+1}|+S_{t+1}] \\
&= K_{\rho} K_{V,t+1}S_t.
\end{align*}
We also note that 
\begin{align*}
C_t&=\esssup_{\Q \in \calQ} \E^\Q_t[(R_t-X_{t+1}-V_{t+1})^+]\\
&\leq |R_{t}|+\esssup_{\Q \in \calQ} \E^\Q_t[|X_{t+1}|+|V_{t+1}|]\\
& \leq (K_{\rho}+1) K_{V,t+1}S_t
\end{align*}
which implies $|V_t|\leq |R_t|+C_t\leq (2K_{\rho}+1)K_{V,t+1}S_t$.
We have proved that \eqref{eq:induc_statement} holds, i.e. the induction step. By the induction principle \eqref{eq:induc_statement} holds for all $t$ and the proof is complete. 
\end{proof}

\section{Construction of sets of probability measures for multiple prior optimal stopping}\label{sec:Q}

Our aim here is to define a useful set $\calQ$ of parametric probability measures that enables the analysis of a wide range of models and that provide solutions to the multiple-prior optimization problem \eqref{Ct-expression}. In particular, the set $\calQ$ constructed below will imply that optimization over $\calQ$ can be reduced to optimization over the set of parameters, see Theorem \ref{thm:general_prior_example} for the precise statement. 

We will define a useful set of probability measures, satisfying all the requirements for applying key results on multiple prior optimal stopping, by defining the corresponding set of density processes $(D_{\lambda,t})_{t=0}^T$ of the form 
\begin{align*}
D_{\lambda,0}:=1, \quad 
D_{\lambda,t}:=\prod_{s=1}^t \int_\Theta f_s(\theta) \lambda_s(\mathrm{d}\theta) \text{ for }
t\in\{1,\dots,T\},
\end{align*}
where $\Theta$ is a set of parameters and $(f_s)_{s=1}^T$ and $(\lambda_s)_{s=1}^T$ are defined below. 

On $(\Omega,\calF_T)$, a probability measure $\Q$ absolutely continuous with respect to $\P$ corresponds to a Radon-Nikodym derivative $D_T\in L^1(\calF_T,\P)$ and together with the filtration $(\calF_t)_{t=0}^T$ give rise to the density process $(D_t)_{t=0}^T$ given by $D_t=\E^{\P}_t[D_T]$.   
Similarly, a set $\calQ$ of probability measures, absolutely continuous with respect to $\P$, corresponds to the set $\calD_T\subset L^1(\calF_T,\P)$ of Radon-Nikodym derivatives. 
Write $\overline{\calD}_T$ for the $L^1$ closure of $\calD_T$ and let $\overline{\calQ}$ be the set of probability measures corresponding to the Radon-Nikodym derivatives $\overline{\calD}_T$. For two probability measures $\Q^{(1)},\Q^{(2)}$ with Radon-Nikodym derivatives $D^{(1)}_T,D^{(2)}_T$ the Radon-Nikodym derivative of the pasting of $\Q^{(1)},\Q^{(2)}$ in $\tau$ is 
$$
D^{(1)}_{\tau}\frac{D^{(2)}_T}{D^{(2)}_{\tau}}.
$$
The following result is both of independent interest and will be relevant for constructing stable sets of probability measures, depending on a parameter, that are useful for multiple prior optimal stopping problems. 

\begin{theorem}\label{lem:stableclosure}
Given $(\Omega,\calF,(\calF_t)_{t=0}^{T},\P)$, let $\calQ$ be a set of probability measures equivalent to $\P$ that is convex and stable under pasting. Let $\calD_T$ be the corresponding set of Radon-Nikodym derivatives and let $\overline{\calD}_T$ be the $L^1$ closure of $\calD_T$. 
Let $\calD_t:=\{D_t=\E^{\P}_t[D_T]:D_T\in \calD_T\}$ and $\overline{\calD}_t:=\{D_t=\E^{\P}_t[D_T]:D_T\in \overline{\calD}_T\}$. 
Then 
\begin{itemize}
\item[(i)] The set $\overline{\calQ}$ corresponding to $\overline{\calD}_T$ is convex and stable under pasting.
\item[(ii)] For each $t$, $\overline{\calD}_t$ is convex and closed in $L^1(\calF_T,\P)$. 
\item[(iii)] If $\calD_T$ is $\P$-uniformly integrable, then for each $t$, $\calD_t$ is weakly relatively compact in $L^1(\calF_T,\P)$ and $\overline{\calD}_t$ is weakly compact in $L^1(\calF_T,\P)$. 
\item[(iv)] If $\calD_{T}$ is $\P$-uniformly integrable, then, for any $\calF_{t+1}$-measurable $\calQ$-uniformly integrable random variable $Y_{t+1}$, 
\begin{align*}
\essinf_{\Q \in \overline{\calQ}}\E^\Q_t\big[Y_{t+1}\big]
=\essinf_{\Q \in \calQ}\E^\Q_t\big[Y_{t+1}\big]
\end{align*}
and similarly with $\essinf$ replaced by $\esssup$. 
\end{itemize}
\end{theorem}

\begin{remark}\label{rem:ui}
Since $\E^{\P}[D]=1$ for $D\in \calD_{T}$ (and similarly for $\overline{\calD}_{T}$), by Lemma 4.10 in Kallenberg \cite{Kallenberg-02}, $\calD_{T}$ is uniformly integrable if 
$$
\lim_{\P(A)\to 0}\sup_{\Q\in \calQ}\Q(A)=0
$$
(and similarly for $\overline{\calD}_{T}$). 
\end{remark}

\begin{proof}[Proof of Theorem \ref{lem:stableclosure}]
For both statements $(i)$ and $(ii)$, it is straightforward to verify that convexity holds so we only prove the remaining claims. 

To prove $(i)$, consider any stopping time $\tau$ and $D^{(1)}_T,D^{(2)}_T\in \overline{\calD}_T$. Take $(D^{(1)}_{n,T})_{n\geq 1},(D^{(2)}_{n,T})_{n\geq 1}\subset \calD_T$ such that $D^{(1)}_{n,T}\to D^{(1)}_T$ and $D^{(2)}_{n,T}\to D^{(2)}_T$ in $L^1$. 
Since $D^{(1)}_{n,T},D^{(2)}_{n,T}\in \calD_T$ and $\calQ$ is stable under pasting, the Radon-Nikodym derivative of the pasting of $D^{(1)}_{n,T},D^{(2)}_{n,T}\in \calD_T$ in $\tau$ is also an element in $\calD_T$. Therefore, statement (i) is proved if we show that there exists a subsequence $(n_i)$ such that 
\begin{align}\label{eq:subseqconvL1}
D^{(1)}_{n_i,\tau}\frac{D^{(2)}_{n_i,T}}{D^{(2)}_{n_i,\tau}} \to D^{(1)}_{\tau}\frac{D^{(2)}_T}{D^{(2)}_{\tau}}
\text{ in } L^1. 
\end{align}
Since, for $k=1,2$, 
\begin{align*}
\E^{\P}[|D^{(k)}_{n,T}-D^{(k)}_{T}|]&=\E^{\P}[\E^{\P}_\tau[|D^{(k)}_{n,T} - D^{(k)}_{T}|]]\\
&\geq \E^{\P}[|\E^{\P}_{\tau}[D^{(k)}_{n,T} - D^{(k)}_{T}]|]\\
&=\E^{\P}[|D^{(k)}_{n,\tau} - D^{(k)}_{\tau}|], 
\end{align*}
we see that $D^{(k)}_{n,\tau}\to D^{(k)}_{\tau}$ in $L^1$. Since convergence in $L^1$ implies convergence in probability which in turn implies a.s. convergence along some subsequence $(n_i)$, we have 
$$
D^{(1)}_{n_i,\tau}\frac{D^{(2)}_{n_i,T}}{D^{(2)}_{n_i,\tau}} \to D^{(1)}_{\tau}\frac{D^{(2)}_T}{D^{(2)}_{\tau}}
\text{ a.s.}, 
$$ 
where we used the fact that $D^{(2)}_{n_i,\tau}$ and $D^{(2)}_{\tau}$ are strictly positive a.s. Since the terms of the sequence on the left-hand side are positive and all have expected values equal to $1$, this sequence is uniformly integrable. Therefore, by Proposition 4.12 in \cite{Kallenberg-02}, the a.s. convergence can be replaced by convergence in $L^1$, i.e.~\eqref{eq:subseqconvL1} holds.   

We now prove $(ii)$. Consider an $L^1$ convergent sequence $(D_{n,t})_{n\geq 1}\subset \overline{\calD}_t$ with limit $D_t$. We will prove that $D_t\in \overline{\calD}_t$. 
Take an arbitrary $D'_T\in \overline{\calD}_T$ and let $D'_t:=\E^{\P}_t[D'_T]$. Set 
\begin{align*}
D_{n,T}:=D_{n,t}\frac{D'_T}{D'_t}, \quad D_{T}:=D_{t}\frac{D'_T}{D'_t}.
\end{align*}
By construction $(D_{n,T})_{n\geq 1}\subset \overline{\calD}_T$ and $D_{n,T}\to D_T$ in $L^1$. Hence, $D_T\in \overline{\calD}_T$. Since $D_t=\E^{\P}[D_T]\in \overline{\calD}_t$ and $D_{n,t}\to D_t$ in $L^1$ the proof of $(ii)$ is complete. 

We now prove $(iii)$. From $(i)$ and $(ii)$ follow that $\overline{\calQ}$ is convex and stable under pasting and, for each $t$, $\overline{\calD}_{t}$ is a convex and closed subset of $L^1(\calF_T,\P)$. A convex and closed subset of $L^1(\calF_T,\P)$ is weakly closed (Theorem A.63 in \cite{Foellmer-Schied-16}). A bounded and uniformly integrable subset of $L^1(\calF_T,\P)$ is weakly relatively compact (Theorem A.70 in \cite{Foellmer-Schied-16}). 
Each $\calD_{t}$ is a bounded subset of $L^1(\calF_T,\P)$:
$$
\sup_{D_t\in \calD_{t}}\Vert D_t\Vert_{L^1}=\sup_{D_t\in \calD_{t}}\E^{\P}[D_t]=1.
$$
Hence, weak relative compactness of $\calD_{t}$ follows if $\calD_{t}$ is uniformly integrable. 
Similarly for $\overline{\calD}_{t}$. 
Moreover, $\calD_{t}$ is uniformly integrable if $\calD_{T}$ is uniformly integrable, as the following argument shows. 
By Lemma 4.10 in \cite{Kallenberg-02}, since $\sup_{D\in \calD_{T}}\E^{\P}[D]=1$, $\calD_{T}$ is $\P$-uniformly integrable if and only if $\lim_{\P(A)\to 0}\sup_{D\in \calD_{T}}\E^{\P}[D;A]=0$. If the latter holds, then in particular
$\lim_{r\to\infty}\sup_{D\in \calD_{T}}\E^{\P}[D;\E^{\P}_t[D]>r]=0$ which is equivalent to 
$$
\lim_{r\to\infty}\sup_{D\in \calD_{T}}\E^{\P}[\E^{\P}_t[D];\E^{\P}_t[D]>r]=0.
$$ 
Hence, $\calD_{t}$ is $\P$-uniformly integrable if $\calD_{T}$ is $\P$-uniformly integrable.
By the same argument, $\overline{\calD}_{t}$ is uniformly integrable if $\overline{\calD}_{T}$ is uniformly integrable. 
However, $\overline{\calD}_{T}$ is uniformly integrable since it is the closure of a uniformly integrable set in $L^1$. 
Hence, $\overline{\calD}_{t}$ is weakly compact in $L^1(\calF_T,\P)$ for every $t$. 
The proof of $(iii)$ is complete.

It remains to prove $(iv)$. 
Notice first that 
\begin{align*}
\essinf_{\Q\in\calQ}\E^{\Q}_t[Y_{t+1}]
=\essinf_{D\in\calD_{T}}\frac{1}{D_t}\E^{\P}_t[DY_{t+1}]
=\essinf_{D\in\calD_{T}}\frac{1}{D_t}\E^{\P}_t[D_{t+1}Y_{t+1}].
\end{align*}
Take $D^*\in\overline{\calD}_{T}$ and $(D_n)\subset \calD_{T}$ with $D_n\to D^*$ in $L^1$. Therefore, 
$D_n\stackrel{\P}{\to}D^*$ which implies $D_nY_{t+1}\stackrel{\P}{\to}D^*Y_{t+1}$. Moreover, $D_{n,t}\stackrel{\P}{\to}D^{*}_t$ since $D_{n,t}=\E^{\P}_t[D_n]$, $D^{*}_{t}=\E^{\P}_t[D^{*}]$ and 
\begin{align*}
\E^{\P}[|\E^{\P}_t[D_n]-\E^{\P}_t[D^{*}]|]\leq \E^{\P}[\E^{\P}_t[|D_n-D^{*}|]]
=\E^{\P}[|D_n-D^{*}|]\to 0
\end{align*}
and convergence in $L^1$ implies convergence in probability. 
Since $Y_{t+1}$ is $\calQ$-uniformly integrable, $\{DY_{t+1}:D\in\calD_{T}\}$ is $\P$-uniformly integrable. Therefore, $D_nY_{t+1}\stackrel{\P}{\to}D^*Y_{t+1}$ implies $D_nY_{t+1}\to D^*Y_{t+1}$ in $L^1$. 
This further implies that $\E^{\P}_t[D_nY_{t+1}]\to \E^{\P}_t[D^*Y_{t+1}]$ in $L^1$ since
\begin{align*}
\E^{\P}[|\E^{\P}_t[D_nY_{t+1}]-\E^{\P}_t[D^*Y_{t+1}]|]
&\leq \E^{\P}[\E^{\P}_t[|D_nY_{t+1}-D^*Y_{t+1}|]]\\
&=\E^{\P}[|D_nY_{t+1}-D^*Y_{t+1}|]\to 0.
\end{align*}
In particular, 
\begin{align*}
\frac{1}{D_{n,t}}\E^{\P}_t[D_nY_{t+1}]\stackrel{\P}{\to} \frac{1}{D^{*}_{t}}\E^{\P}_t[D^{*}Y_{t+1}]
\end{align*}
which implies that there exists a subsequence $(n_i)$ such that 
\begin{align*}
\frac{1}{D_{n_i,t}}\E^{\P}_t[D_{n_i}Y_{t+1}]\stackrel{\text{a.s.}}{\to} \frac{1}{D^{*}_{t}}\E^{\P}_t[D^{*}Y_{t+1}]. 
\end{align*}
Therefore, for any $D^*\in\overline{\calD}_{T}$, 
\begin{align*}
\frac{1}{D^{*}_t}\E^{\P}_t[D^*Y_{t+1}]\geq \essinf_{D\in \calD_{T}}\frac{1}{D_t}\E^{\P}_t[DY_{t+1}].
\end{align*}
The same argument shows the corresponding identity for $\esssup$. The proof is complete.  
\end{proof}

Consider a parameter set $\Theta$ which is taken to be a subset of a complete and separable metric space. 
For each $t\in\{1,\dots,T\}$, let $f_t\geq 0$ be a measurable function on $\Omega \times \Theta$ such that $\omega \mapsto f_t(\omega,\theta)$ is $\calF_t$-measurable for each $\theta\in\Theta$. 
It is assumed that the $\calF_t$-measurable random variables $f_t(\theta)$ satisfy 
\begin{align}
& \essinf_{\theta\in\Theta}f_t(\theta)>0 \; \P\text{-a.s.~for } t \in\{1,\dots,T\}, \label{eq:UQcond1} \\
& \E^{\P}_{t-1}[f_t(\theta)]=1 \text{ for } (t,\theta) \in\{1,\dots,T\} \times \Theta. \label{eq:UQcond3}
\end{align}
For each $t\in\{1,\dots,T\}$, let $\lambda_t$ be an $\calF_{t-1}$-measurable random element in the space $\calP(\Theta)$ of probability measures on $\Theta$ equipped with the topology of weak convergence. Let $\Lambda_t$ be the set of all such random probability measures. 
For $\lambda_t\in\Lambda_t$ for all $t$, let 
\begin{align}\label{eq:D_lambda}
D_{\lambda,T}:=\prod_{t=1}^T \int_\Theta f_t(\theta) \lambda_t(\mathrm{d}\theta), 
\quad D_{\lambda,t}:=\E^{\P}_t[D_{\lambda,T}] \text{ for } t<T. 
\end{align}
Notice that, due to properties \eqref{eq:UQcond1} and \eqref{eq:UQcond3}, $(D_{\lambda,t})_{t=0}^T$ is a positive $\P$-martingale with $\E^{\P}[D_{\lambda,t}]=1$. 
Let 
\begin{align*}
\calD_{f,T} := \bigg\{\prod_{t=1}^T f_t(\theta):\theta\in\Theta\bigg\}, \quad 
\widetilde{\calD}_{f,T} := \Big\{D_{\lambda,T}:\lambda_t\in \Lambda_t \text{ for all } t\Big\} 
\end{align*}
and let $\overline{\calD}_{f,T}$ be the $L^1$-closure of $\calD_{f,T}$. For $t=0,1,\dots,T-1$, let 
\begin{align*}
\calD_{f,t} := \Big\{D_t=\E^{\P}_t[D_T]:D_T\in \calD_{f,T}\Big\}  
\end{align*}
and let $\widetilde{\calD}_{f,t}$ and $\overline{\calD}_{f,t}$ be defined analogously. 

\begin{definition}\label{def:curlyQs}
Denote by $\calQ_{\Theta},\widetilde{\calQ}_{\Theta},\overline{\calQ}_{\Theta}$ the sets of probability measures corresponding to the sets $\calD_{f,T},\widetilde{\calD}_{f,T},\overline{\calD}_{f,T}$ of Radon-Nikodym derivatives with respect to $\P$. 
\end{definition}

Notice that $\calQ_{\Theta}$ corresponds to only considering measures $\lambda_t(\cdot)=1_{\{\theta\}}(\cdot)$ in \eqref{eq:D_lambda}. 
We will show in Theorem \ref{thm:general_prior_example} that $\overline{\calQ}_{\Theta}$ has the properties assumed in Theorem \ref{Vt-def-thm}. We also show that Theorem \ref{Vt-def-thm} holds also for $\widetilde{\calQ}_{\Theta}$. 

\begin{theorem}\label{thm:general_prior_example}
Consider the sets $\calQ_{\Theta},\widetilde{\calQ}_{\Theta},\overline{\calQ}_{\Theta}$ and $\calD_{f,T},\widetilde{\calD}_{f,T},\overline{\calD}_{f,T}$ in Definition \ref{def:curlyQs}. 
\begin{itemize}
\item[(i)] The sets $\widetilde{\calQ}_{\Theta}$ and $\overline{\calQ}_{\Theta}$ are convex and stable under pasting. 
\item[(ii)] For every $t\in\{1,\dots,T\}$, $\overline{\calD}_{f,t}$ is closed in $L^1(\calF_T,\P)$. 
\item[(iii)] If $\widetilde{\calD}_{f,T}$ is $\P$-uniformly integrable, then $\widetilde{\calD}_{f,t}$ is weakly relatively compact in $L^1(\calF_T,\P)$ and $\overline{\calD}_{f,t}$ is weakly compact in $L^1(\calF_T,\P)$ for every $t\in\{0,\dots,T\}$. 
\item[(iv)] If $\widetilde{\calD}_{f,T}$ is $\P$-uniformly integrable, then, for any $\calF_{t+1}$-measurable $\widetilde{\calQ}_{\Theta}$-uniformly integrable random variable $Y_{t+1}$, 
\begin{align*}
\essinf_{\Q \in \overline{\calQ}_{\Theta}}\E^\Q_t\big[Y_{t+1}\big]
&=\essinf_{\Q \in \widetilde{\calQ}_{\Theta}}\E^\Q_t\big[Y_{t+1}\big]\\
&=\essinf_{\Q \in \calQ_{\Theta}}\E^\Q_t\big[Y_{t+1}\big]\\
&=\essinf_{\theta  \in \Theta} \E^\P_t\big[Y_{t+1} f_{t+1}(\theta)\big]
\end{align*}
and similarly with $\essinf$ replaced by $\esssup$. 
\end{itemize}
\end{theorem}

\begin{proof}[Proof of Theorem \ref{thm:general_prior_example}]
We first prove $(i)$. We first prove convexity and stability under pasting for the set of probability measures with Radon-Nikodym derivatives $D_T\in \widetilde{\calD}_{f,T}$ with respect to $\P$. 
We first prove convexity. 
Note that for a density process $(D_t)_{t=0}^T$ with $D_T\in\widetilde{\calD}_{f,T}$, 
\begin{align*}
\frac{D_{t+1}}{D_t}=\frac{\prod_{s=1}^{t+1}\int_{\Theta}f_s(\theta)\lambda_s(\rmd \theta)}{\prod_{s=1}^{t}\int_{\Theta}f_s(\theta)\lambda_s(\rmd \theta)}
=\int_{\Theta}f_{t+1}(\theta)\lambda_{t+1}(\rmd \theta).
\end{align*}
Consider density processes 
$D^{(1)},D^{(2)}$ with $D^{(1)}_T,D^{(2)}_T\in \widetilde{\calD}_{f,T}$, let $c\in (0,1)$ and set $D^{(3)}:= cD^{(1)}+(1-c)D^{(2)}$.
Then 
\begin{align*}
\frac{D^{(3)}_{t+1}}{D^{(3)}_t}
&=\frac{1}{D^{(3)}_t}\Big(cD^{(1)}_{t} \int_{\Theta}f_{t+1}(\theta)\lambda^{(1)}_{t+1}(\rmd \theta)+(1-c)D^{(2)}_{t} \int_{\Theta}f_{t+1}(\theta)\lambda^{(2)}_{t+1}(\rmd \theta)\Big)\\
&=\int_{\Theta}f_{t+1}(\theta)\bigg(\frac{1}{D^{(3)}_t}\Big(cD^{(1)}_t\lambda^{(1)}_{t+1}+(1-c)D^{(2)}_t\lambda^{(2)}_{t+1}\Big)(\rmd \theta)\bigg)
\end{align*}
and the convexity property follows since 
\begin{align*}
\frac{1}{D^{(3)}_{t}}\Big(cD^{(1)}_{t}\lambda^{(1)}_{t+1}+(1-c)D^{(2)}_{t}\lambda^{(2)}_{t+1}\Big)\in \Lambda_{t+1}.
\end{align*}
We now prove stability under pasting. 
Consider density processes 
$D^{(1)},D^{(2)}$ with $D^{(1)}_T,D^{(2)}_T\in \widetilde{\calD}_{f,T}$, and let $\tau$ be a stopping time. Then 
\begin{align*}
D^{(3)}_t&:=\prod_{s=1}^t\bigg(\indic\{s\leq \tau\}\frac{D^{(1)}_s}{D^{(1)}_{s-1}}+\indic\{s> \tau\}\frac{D^{(2)}_s}{D^{(2)}_{s-1}}\bigg)\\
&=\prod_{s=1}^t \int_{\Theta}f_s(\theta)\Big(\indic\{s\leq \tau\}\lambda^{(1)}_s+\indic\{s> \tau\}\lambda^{(2)}_s\Big)(\rmd \theta)
\end{align*}
and since $\indic\{s\leq \tau\},\indic\{s> \tau\}$ are $\calF_{s-1}$-measurable, 
\begin{align*}
\indic\{s\leq \tau\}\lambda^{(1)}_s+\indic\{s>\tau\}\lambda^{(2)}_s \in \Lambda_s,
\end{align*}
which proves stability under pasting. Theorem \ref{lem:stableclosure}$(i)$ completes the proof of $(i)$. 

Notice that $(ii)$ follows immediately from $(i)$ together with Theorem \ref{lem:stableclosure}$(ii)$. 
Similarly, $(iii)$ follows immediately from $(i)$ together with Theorem \ref{lem:stableclosure}$(iii)$.  

It remains to prove $(iv)$. 
The two last identities in $(iv)$ follow from the definition of $\calQ_{\Theta}$ and $\widetilde{\calQ}_{\Theta}$:
\begin{align*}
\essinf_{\Q \in \widetilde{\calQ}_{\Theta}}\E^\Q_t\big[Y_{t+1}\big]
&=\essinf_{\lambda_{t+1} \in \Lambda_{t+1}}\E^\P_t\bigg[Y_{t+1}\int_{\Theta}f_{t+1}(\theta)\lambda_{t+1}(\mathrm{d} \theta)\bigg]\\
&=\essinf_{\theta  \in \Theta} \E^\P_t \big[Y_{t+1} f_{t+1}\big(\theta\big)\big].
\end{align*}
The first identity follows from Theorem \ref{lem:stableclosure}$(iv)$. The proof is complete.    
\end{proof}

\section{Gaussian example}\label{sec:gaussian_example}

In this section we consider an application illustrating the general theory presented up to this point. 
We consider a setting where the liability cash flow is Gaussian, both under $\P$ and under any $\Q_{\theta}\in\calQ_{\Theta}$. 
We consider two cases:
\begin{itemize}
\item[Case 1]
In this case we assume that $\calQ=\calQ_{\Theta}$. Recall that $\calQ_{\Theta}$ is not stable under pasting: this decision maker does not consider probability measures corresponding to switching between probability measures in $\calQ_{\Theta}$ depending on information revealed over time. 
\item[Case 2]
In this case we assume that $\calQ=\widetilde{\calQ}_{\Theta}$. This decision maker exhibits a behaviour that is time consistent. Notice that $\widetilde{\calQ}_{\Theta}$ is considerably larger than $\calQ_{\Theta}\subset \widetilde{\calQ}_{\Theta}$. 
\end{itemize}

The liability cash flow is assumed to be fully nonhedgeable by financial assets and consequently we take $X^{r}=0$ which means that $X=X^{o}$. In order to make the illustration clear, we choose $T=2$. Let $C_{i,k}:=C_{i,k}^{\text{orig}}/v_i$ denote the exposure adjusted cumulative amount of payments to policyholders for accident year $i$, where $v_i$ is a known exposure measure for accident year $i$. 
The evolution of the exposure adjusted cumulative amounts is assumed to satisfy 
\begin{align*}
C_{i,1}=\beta_0^{\P}+\frac{\sigma_0^{\P}}{\sqrt{v_i}}\varepsilon_{i,1}, \quad 
C_{i,2}=\beta_1^{\P}C_{i,1}+\frac{\sigma_1^{\P}}{\sqrt{v_i}}\varepsilon_{i,2}, 
\end{align*}
where all $\varepsilon_{i,k}$ are independent and $N(0,1)$ with respect to $\P$. Suppose that we are uncertain about the parameter values and want to consider probability measures $\Q_{\theta}$, $\theta=(\beta_0,\sigma_0,\beta_1,\sigma_1)$, such that 
\begin{align*}
C_{i,1}=\beta_0+\frac{\sigma_0}{\sqrt{v_i}}\varepsilon_{i,1}^{\theta}, \quad
C_{i,2}=\beta_1C_{i,1}+\frac{\sigma_1}{\sqrt{v_i}}\varepsilon_{i,2}^{\theta}, 
\end{align*}
where all $\varepsilon_{i,k}^{\theta}$ are independent and $N(0,1)$ with respect to $\Q_{\theta}$. We choose a parameter set $\Theta\subset (0,\infty)^{4}$ that describes the uncertainty about the parameter values. 

Suppose that $C_{i,k}$ with $i+k\leq 0$ are observed at time $0$ and that $C_{i,k}$ with $i+k=t$, $t=1,2$, are observed at times $t>0$ and therefore contain cash flows that are part of the outstanding liability to the policyholders. 
The (incremental) liability cash flow $X=(X_1,X_2)$ is given by 
\begin{align*}
X=\bigg(v_{-1}(C_{-1,2}-C_{-1,1})+v_0C_{0,1},v_0(C_{0,2}-C_{0,1})\bigg).
\end{align*}
Notice that $C_{-1,1}$ is here considered to be a known constant. 
Direct computations give  
\begin{align*}
\E^{\P}[X_1+X_2]&=v_{-1}(\beta_1^{\P}-1)C_{-1,1}+v_0\beta_0^{\P}\beta_1^{\P}, \\ 
\E^{\Q_{\theta}}[X_1+X_2]&=v_{-1}(\beta_1-1)C_{-1,1}+v_0\beta_0\beta_1.
\end{align*}
The filtration is given by the $\sigma$-algebras $\calF_0=\{\emptyset,\Omega\}$, $\calF_1=\sigma(\varepsilon_{-1,2},\varepsilon_{0,1})$ and $\calF_2=\sigma(\varepsilon_{0,2}) \vee \calF_1$. 
In order to have the correct evolution with respect to $\Q_{\theta}$ of the cumulative amounts it is seen that we must require that 
\begin{align*}
&\Q_{\theta}(\varepsilon_{-1,2}\in \cdot\mid \calF_{0})\sim N(\mu_{-1,2},\sigma_{-1,2}^2),\quad 
\mu_{-1,2}=\frac{\beta_1-\beta_1^{\P}}{\sigma_{1}^{\P}/\sqrt{v_{-1}}}C_{-1,1}, \,
\sigma_{-1,2}=\frac{\sigma_{1}}{\sigma_{1}^{\P}}, \\
&\Q_{\theta}(\varepsilon_{0,1}\in \cdot\mid \calF_{0})\sim N(\mu_{0,1},\sigma_{0,1}^2),\quad 
\mu_{0,1}=\frac{\beta_0-\beta_0^{\P}}{\sigma_{0}^{\P}/\sqrt{v_0}}, \,
\sigma_{0,1}=\frac{\sigma_{0}}{\sigma_{0}^{\P}}, \\
&\Q_{\theta}(\varepsilon_{0,2}\in \cdot\mid \calF_{1})\sim N(\mu_{0,2},\sigma_{0,2}^2),\quad 
\mu_{0,2}=\frac{\beta_1-\beta_1^{\P}}{\sigma_{1}^{\P}/\sqrt{v_0}}C_{0,1}, \,
\sigma_{0,2}=\frac{\sigma_{1}}{\sigma_{1}^{\P}}.
\end{align*}
This corresponds to, in the setting of Section \ref{sec:Q}, choosing 
\begin{align}
f_1(\theta)&=\frac{\phi(\varepsilon_{-1,2};\mu_{-1,2},\sigma_{-1,2}^2)\phi(\varepsilon_{0,1};\mu_{0,1},\sigma_{0,1}^2)}{\phi(\varepsilon_{-1,2};0,1)\phi(\varepsilon_{0,1};0,1)}, \label{eq:f1}\\
f_2(\theta)&=\frac{\phi(\varepsilon_{0,2};\mu_{0,2},\sigma_{0,2}^2)}{\phi(\varepsilon_{0,2};0,1)}, \nonumber
\end{align}
where $\varphi(x;\mu,\sigma^2)$ denotes the density function of $N(\mu,\sigma^2)$. 
By Remark \ref{rem:ui}, the set $\widetilde{\calD}_{f,2}$ is $\P$-uniformly integrable if 
$$
\lim_{\P(A)\to 0}\sup_{\Q\in \widetilde{\calQ}_{\Theta}}\Q(A)=0
$$
which holds here since $\sigma_k^{\P}/\sigma_k$ and $|\beta_k-\beta_k^{\P}|$ both take values in bounded intervals bounded away from $0$. The sets $A\in\calF_2$ are of type $\{(\varepsilon_{-1,2},\varepsilon_{0,1},\varepsilon_{0,2})\in B\}$ for measurable sets $B\subset \R^3$ such that $\P((\varepsilon_{-1,2},\varepsilon_{0,1},\varepsilon_{0,2})\in B)\to 0$. 
Therefore, it follows from Theorem \ref{thm:general_prior_example} that 
the set $\widetilde{\calQ}_{\Theta}$ in Definition \ref{def:curlyQs} satisfies the requirements for multiple prior optimal stopping. In particular, Theorem \ref{Vt-def-thm} holds with $\calQ=\widetilde{\calQ}_{\Theta}$.  

$\Theta$ can be chosen to reflect parameter uncertainty. 
To illustrate how such a choice may be implemented, consider the regression estimators from Lindholm et al.\cite{lindholm2017valuation} based on data from accident years $i=i_0,\dots,-1$:
\begin{align*}
\widehat{\beta_0^{\P}}&=\frac{\sum_{i=i_0}^{-1}v_iC_{i,1}}{\sum_{i=i_0}^{-1}v_i}, \quad 
\widehat{(\sigma_0^{\P})^2}=\frac{1}{-i_0-1}\sum_{i=i_0}^{-1}v_i(C_{i,1}-\widehat{\beta_0^{\P}})^2, \\
\widehat{\beta_1^{\P}}&=\frac{\sum_{i=i_0}^{-2}v_iC_{i,1}C_{i,2}}{\sum_{i=i_0}^{-2}v_iC_{i,1}^2}, \quad 
\widehat{(\sigma_1^{\P})^2}=\frac{1}{-i_0-2}\sum_{i=i_0}^{-2}v_i(C_{i,2}-\widehat{\beta_1^{\P}}C_{i,1})^2.
\end{align*}
Here $i_0$ denotes the index of the first accident year observed. These estimators are unbiased and uncorrelated. We now proceed with a numerical illustration, with parameter values $(\beta_0^{\P}, \sigma_0^{\P}, \beta_1^{\P}, \sigma_1^{\P})= (2/3,1/5,3/2,1/5)$, $i_0=-10$, and $v_i=1$ for $i=-10,\dots,0$. Based on these parameter values and a large number $n$ of simulated independent standard normal $\varepsilon_{i,j}$, leading to  $n$ iid copies of $C_{-10,1},\dots,C_{-1,1},C_{-10,2},\dots,C_{-2,2}$, we estimate $(\beta_0^{\P}, \sigma_0^{\P}, \beta_1^{\P}, \sigma_1^{\P})$ $n$ times. 
Figure \ref{fig:scatter_parameterestimates} presents scatter plots, which suggests that the iid vectors of estimators are approximately $N_{4}(\mu,\Sigma)$-distributed, where $\mu$ and $\Sigma$ are the sample mean and sample covariance matrix. We can therefore shape an approximative  confidence region with confidence level $p$ of the parameter values by the squared Mahalanobis distance as
\begin{align*}
\Theta&=\Big\{z\in\R^4:(z-\mu)^{\trans}\Sigma^{-1}(z-\mu)\leq F_{\chi^2(4)}^{-1}(p)\Big\}\\
&=\Big\{\mu+rLs\in\R^4:r^2\leq F_{\chi^2(4)}^{-1}(p),\,s\in\mathbb{S}^3\Big\},
\end{align*}  
where $L$ is the (lower triangular) Cholesky decomposition of $LL^{\trans}=\Sigma$, $F_{\chi^2(4)}$ is the distribution function of the $\chi^2(4)$ and $\mathbb{S}^3$ is the unit sphere in $\R^4$. For the evaluation at time 1, only $(\beta_1, \sigma_1)$ needs to be considered, leading to a set $\Theta_{\beta_1,\sigma_1}\subset \R^2$ satisfying that 
$$
\{(0,0,\beta_1,\sigma_1):(\beta_1,\sigma_1)\in \Theta_{\beta_1,\sigma_1}\}
$$ 
is the orthogonal projection of $\Theta$ onto the $(\beta_1, \sigma_1)$ coordinate plane: $\beta_0=\sigma_0=0$. Explicitly, 
\begin{align*}
\Theta_{\beta_1,\sigma_1}&=\Big\{z\in\R^2:(z-\mu_{\beta_1,\sigma_1})^{\trans}\Sigma_{\beta_1,\sigma_1}^{-1}(z-\mu_{\beta_1,\sigma_1})\leq F_{\chi^2(4)}^{-1}(p)\Big\}\\
&=\Big\{\mu_{\beta_1,\sigma_1}+rL_{\beta_1,\sigma_1}s\in\R^2:r^2\leq F_{\chi^2(4)}^{-1}(p),\,s\in\mathbb{S}^1\Big\},
\end{align*}
where $\mathbb{S}^1$ is the unit sphere in $\R^2$, $\mu_{\beta_1,\sigma_1}$ is the subvector of the last two entries of $\mu$ and $L_{\beta_1,\sigma_1}$ is the Cholesky decomposition of the submatrix $\Sigma_{\beta_1,\sigma_1}$ of $\Sigma$. 
Similarly, to compute the upper bound of $V_0$ in \eqref{eq:V0_bound_gauss}, only $(\beta_0, \beta_1)$ need to be considered, leading to a similar set $\Theta_{\beta_0,\beta_1}\subset \R^2$.  

The left plot in Figure \ref{fig:scatter_parameterestimates} shows a scatted plot of $1000$ iid estimates of $(\beta_0^{\P},\beta_1^{\P})$ together with boundaries $\partial \Theta_{\beta_0,\beta_1}$ for $p=0.1$ (blue) and for $p=0.9$ (red). 
 The right plot in Figure \ref{fig:scatter_parameterestimates} shows a scatted plot of $1000$ iid estimates of $(\beta_1^{\P},\sigma_1^{\P})$ together with boundaries $\partial \Theta_{\beta_1,\sigma_1}$ for $p=0.1$ (blue) and for $p=0.9$ (red). 

\begin{figure}[ht]\center
\includegraphics[scale=0.30]{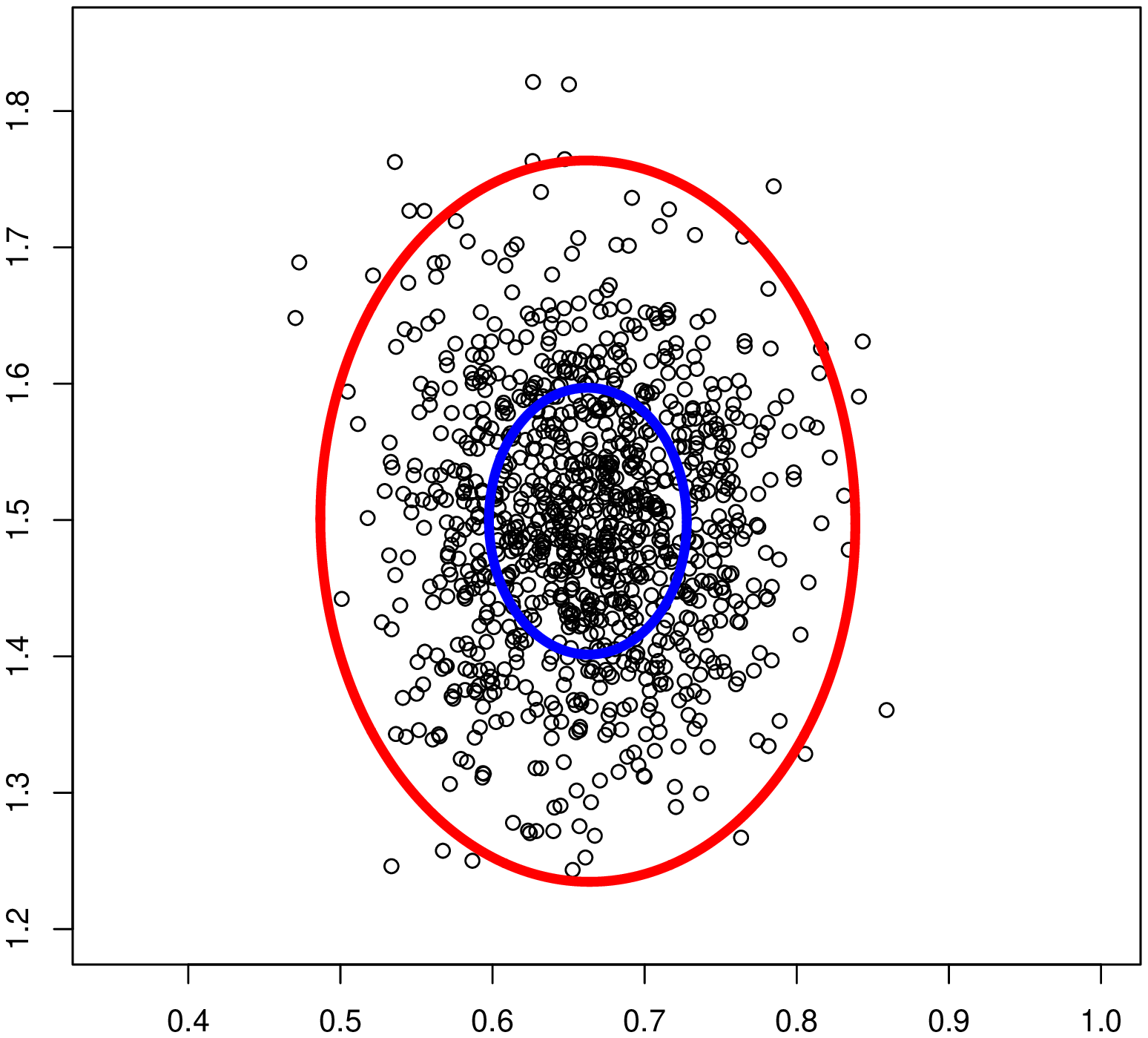}
\includegraphics[scale=0.30]{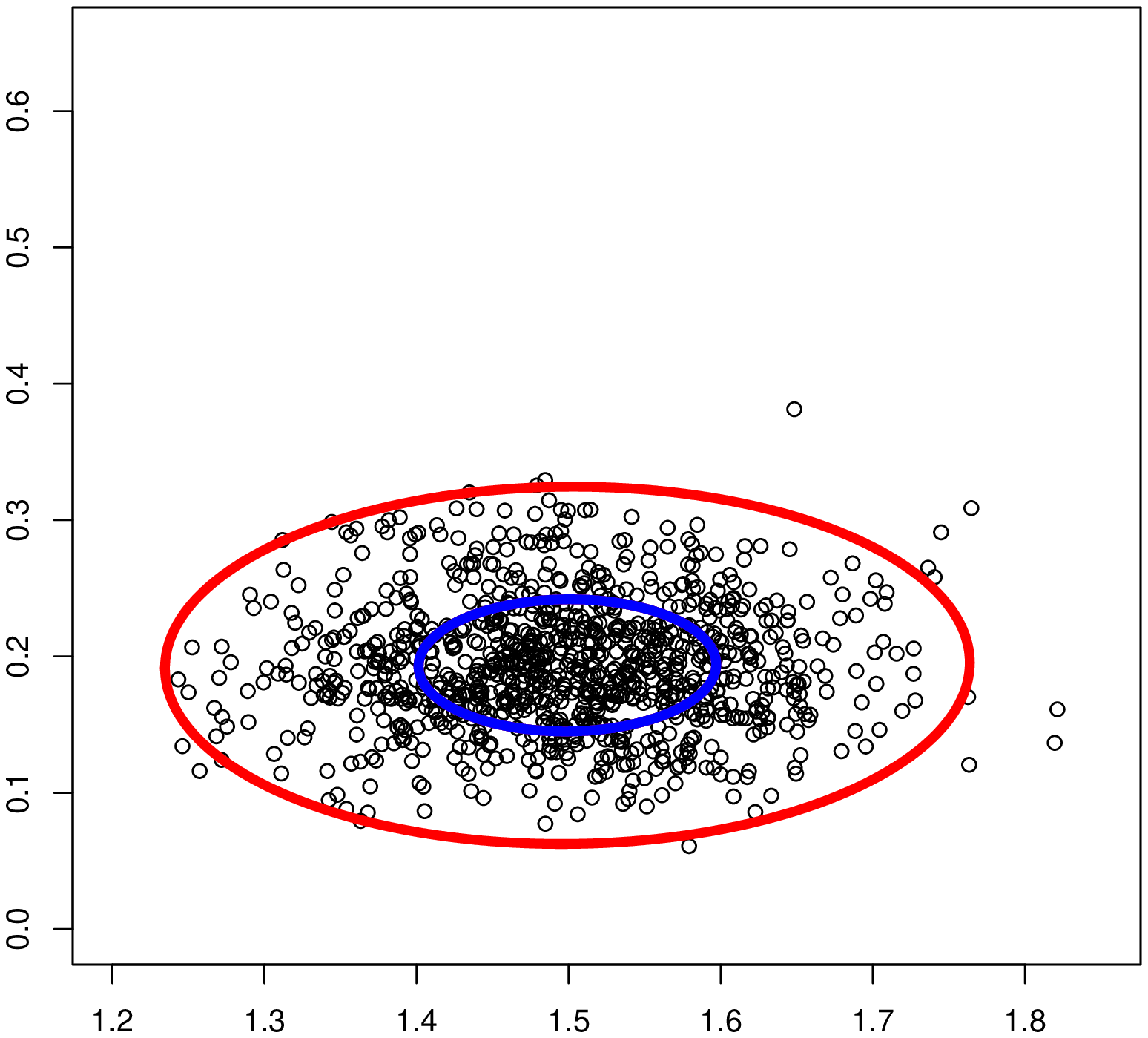}
\caption{Scatter plots of $1000$ iid estimates of $(\beta_0^{\P},\beta_1^{\P})$ (left) and of $(\beta_1^{\P},\sigma_1^{\P})$ (right), together with boundaries of the parameter regions $\Theta_{\beta_0,\beta_1}$ (left) and $\Theta_{\beta_1,\sigma_1}$ (right) for $p=0.1$ (blue) and for $p=0.9$ (red).}
\label{fig:scatter_parameterestimates}
\end{figure}

Let $\rho_0,\rho_1$ be conditional monetary risk measures defined in terms of conditional quantiles with respect to $\P$, such as, for $t=0,1$, $\rho_t=\VaR_{t,p}$ or $\rho_t=\AVaR_{t,p}$. In both cases, $c:=\rho_0(e^{\P}_{1})=\rho_1(e^\P_{2})$ is a constant for an $\calF_{t+1}$-measurable $e^{\P}_{t}\sim N(0,1)$ and independent of $\calF_t$ with respect to $\P$.   
Then 
\begin{align*}
R_1&=\rho_1(-X_2)
=\rho_1(-\E^\P_1[X_{2}]+\Var^\P_1(X_{2})^{1/2}e^{\P}_{2})
=\E^\P_1[X_{2}]+\Var^\P_1(X_{2})^{1/2}c\\
&=v_0(\beta_1^{\P}-1)C_{0,1}+\sqrt{v_0}\sigma_1^{\P}c.
\end{align*}

\subsection{Case 1: computing upper and lower bounds for $V_0$}

In this case $\calQ=\calQ_{\Theta}$ does not satisfy the conditions of Theorem \ref{Vt-def-thm} and therefore we can not compute $V_0$ by backward recursion. However, upper and lower bounds for $V_0$ are easily computed. 
From \eqref{eq:V0generalupperbound} we have the upper bound 
\begin{align}
V_0&\leq \sup_{\Q\in\calQ_{\Theta}}\E^{\Q}[X_1+X_2] \nonumber \\
&=\sup\Big\{v_{-1}(\beta_1-1)C_{-1,1}+v_0\beta_0\beta_1:(\beta_0,\sigma_0,\beta_1,\sigma_1)\in\Theta\Big\}=:\overline{V}_0. \label{eq:V0_bound_gauss}
\end{align}
From \eqref{eq:V0generallowerbound} we have the lower bound 
\begin{align*}
V_0&\geq \sup_{\Q \in \calQ_{\Theta}}\inf_{\tau \in \mathcal{S}_{1,T+1}}\E_0^{\Q}\bigg[\sum_{s=1}^{\tau-1}X_s+R_{\tau-1}\bigg] 
=: \underline{V}_0.
\end{align*}
In the setting of Section \ref{sec:Q}, for each $\theta\in\Theta$, with $V_T^{\theta}=R_T^{\theta}=0$, we solve the backward recursion 
\begin{align*}
R_t^{\theta}&=\rho_t(-X_{t+1}-V_{t+1}^{\theta}), \\
V_t^{\theta}&=R_{t}^{\theta}-\E^{\Q_{\theta}}_{t}[(R_{t}^{\theta}-X_{t+1}-V_{t+1}^{\theta})^{+}], 
\end{align*}
and then compute  
\begin{align*}
\underline{V}_0=\sup_{\theta \in \Theta}V_0^{\theta}. 
\end{align*}
Notice that $V_t^{\theta},R_t^{\theta}$ corresponds to the quantities $V_t,R_t$ in the special case $\calQ=\{\Q_{\theta}\}$. 
Computing $\underline{V}_0$ is simpler than computing $V_0$ since the former involves just one optimisation over the parameter set $\Theta$ rather than $T$ nested optimisations for the latter. 

We now demonstrate how $\underline{V}_0$ is computed in the current Gaussian setting. As shown above $R_1^{\theta}=v_0(\beta_1^{\P}-1)C_{0,1}+\sqrt{v_0}\sigma_1^{\P}c$ (which does not depend on $\theta$) and 
\begin{align*}
C_1^{\theta}=\E^{\Q_{\theta}}_1\big[\big(\rho_1(-X_{2})-X_{2}\big)^+\big]
=\E^{\Q_{\theta}}_1[\big(a(\theta,C_{0,1})-b(\theta)e^{\theta}_{2}\big)^+], 
\end{align*}
where $e^{\theta}_{2}\sim N(0,1)$ with respect to $\Q_{\theta}$ and independent of $\calF_1$, and 
\begin{align*}
a(\theta,C_{0,1})&=-\E^{\Q_{\theta}}_1[X_{2}]+\E^\P_1[X_{2}]+\Var^\P_1(X_{2})^{1/2} \rho_1(e^\P_{2})\\
&=v_0(\beta_1^{\P}-\beta_1)C_{0,1}+\sqrt{v_0}\sigma_1^{\P}c, \\
b(\theta)&=\Var^{\Q_{\theta}}_1(X_{2})^{1/2} = \sqrt{v_0}\sigma_1. 
\end{align*} 
Straightforward calculations show that 
\begin{align*}
\E^{\Q_{\theta}}_1[\big(a(\theta,C_{0,1})-b(\theta)e^{\theta}_{2}\big)^+]
&=a(\theta,C_{0,1})\Phi\bigg(\frac{a(\theta,C_{0,1})}{b(\theta)}\bigg)+b(\theta)\phi\bigg(\frac{a(\theta,C_{0,1})}{b(\theta)}\bigg)\\
&=:g(\theta,C_{0,1}).
\end{align*}
Consequently, 
\begin{align*}
X_1+V_1^{\theta}&=X_1+R_1^{\theta}-C_1^{\theta}\\
&=v_{-1}(\beta_1-1)C_{-1,1}+\sqrt{v_{-1}}\sigma_1\varepsilon_{-1,2}^{\theta}+\sqrt{v_0}\sigma_1^{\P}c\\
&\quad+\beta_1^{\P}\bigg(v_0\beta_0+\sqrt{v_0}\sigma_1\varepsilon_{0,1}^{\theta}\bigg)
-g\bigg(\theta,v_0\beta_0+\sqrt{v_0}\sigma_1\varepsilon_{0,1}^{\theta}\bigg)\\
&=v_{-1}(\beta_1^{\P}-1)C_{-1,1}+\sqrt{v_{-1}}\sigma_1^{\P}\varepsilon_{-1,2}+\sqrt{v_0}\sigma_1^{\P}c\\
&\quad+\beta_1^{\P}\bigg(v_0\beta_0^{\P}+\sqrt{v_0}\sigma_1^{\P}\varepsilon_{0,1}\bigg)
-g\bigg(\theta,v_0\beta_0^{\P}+\sqrt{v_0}\sigma_1^{\P}\varepsilon_{0,1}\bigg)
\end{align*}
from which $R_0^{\theta}=\rho_0(-X_1-V_1^{\theta})$ can be estimated with arbitrary accuracy by simulating iid copies of $X_1+V_1^{\theta}$ with respect to $\P$ and computing the empirical estimate, and $C_0^{\theta}=\E^{\Q_{\theta}}[(R_0^{\theta}-X_1-V_1^{\theta})^{+}]$ can be estimated similarly by simulating iid copies with respect to $\Q_{\theta}$ and approximating the expectation by the empirical mean. Finally, 
\begin{align*}
\underline{V}_0=\sup_{\theta\in \Theta}\big(R_0^{\theta}-C_0^{\theta}\big)
=\sup_{\theta\in \partial \Theta}\big(R_0^{\theta}-C_0^{\theta}\big).
\end{align*}
Table \ref{tab:V0lowerupperbounds} shows numerical values for lower bounds $\underline{V}_0$ and for upper bounds $\overline{V}_0$. These values are based on 
$v_{-1}=v_0=1$, $C_{-1,1}=\beta_0^{\P}$, 
$\rho_t=\VaR_{t,q}$ with $q=0.005,0.01,0.05,0.10$ and parameters sets $\Theta$ of varying size corresponding to $r^2\leq F_{\chi^2(4)}^{-1}(p)$ with $p=0.1,0.5,0.9$. The main message of Table \ref{tab:V0lowerupperbounds} is that the intervals $(\underline{V}_0,\overline{V}_0)$ are very narrow for $q$ small and therefore the upper bound $\overline{V}_0$ is an accurate estimate of $V_0$ when $q$ is small. Notice that the upper bound is both easily computed and has attractive theoretical properties.  

\subsection{Case 2: computing $V_0$ and an upper bound for $V_0$}

In this case $\calQ=\widetilde{\calQ}_{\Theta}$ and the general lower bound $\underline{V}_0$ coincides with $V_0$ and therefore its computation by backward recursion is somewhat involved. However, the upper bound is still fairly straightforward to compute. Notice that the lower bound computed for Case 1 is a lower bound for $V_0$ in the current Case 2 since $\calQ_{\Theta}\subset \widetilde{\calQ}_{\Theta}$. 

We begin by computing the upper bound using the law of iterated expectations, extended to the multiple prior setting, and Theorem \ref{thm:general_prior_example}: 
\begin{align*}
\overline{V}_0&=\sup_{\Q\in\widetilde{\calQ}_{\Theta}}\E^{\Q}_0[X_1+X_2]\\
&=\sup_{\Q\in\widetilde{\calQ}_{\Theta}}\E^{\Q}_0[X_1+\esssup_{\Q'\in\widetilde{\calQ}_{\Theta}}\E^{\Q'}_1[X_2]]\\
&=\sup_{\Q\in\calQ_{\Theta}}\E^{\Q}_0[X_1+\esssup_{\Q'\in\calQ_{\Theta}}\E^{\Q'}_1[X_2]].
\end{align*}
Notice that, with $\beta_{1,\min}>1$,  
\begin{align*}
\esssup_{\Q'\in\calQ_{\Theta}}\E^{\Q'}_1[X_2]
&=v_0(\beta_{1,\max}-1)C_{0,1}\indic_{\{C_{0,1}\geq 0\}}+v_0(\beta_{1,\min}-1)C_{0,1}\indic_{\{C_{0,1}<0\}},
\end{align*}
where $\beta_{1,\max}:=\max\{\beta_1:(\beta_0,\sigma_0,\beta_1,\sigma_1)\in \Theta\}$ and similarly for $\beta_{1,\min}$. 
Therefore, 
\begin{align*}
\overline{V}_0
=\sup_{(\beta_0,\sigma_0,\beta_1,\sigma_1)\in \Theta}\Big(&v_{-1}(\beta_1-1)C_{-1,1}+v_0\beta_0\\
&+v_0(\beta_{1,\max}-1)\big(\beta_0+\sigma_0\Phi(-\beta_0/\sigma_0)\big)\\
&-v_0(\beta_{1,\min}-1)\sigma_0\Phi(-\beta_0/\sigma_0)\Big)
\end{align*}
$R_1$ is calculated explicitly as above.   
Computing $C_1$ means computing 
$$
C_1=\esssup_{\theta_1\in \partial \Theta_{\beta_1,\sigma_1}}g(\theta_1,C_{0,1}),
$$
where, with some abuse of notation, we consider $g$ to be defined for parameters $\theta_1\in \Theta_{\beta_1,\sigma_1}$ 
rather than $\theta\in \Theta$. In practice, this means determining a function $h:\R\to\R$ such that 
$h(C^k_{0,1})=\max_{\theta_1\in \partial \Theta_{\beta_1,\sigma_1}}g(\theta_1,C^k_{0,1})$ for suitably many simulated iid copies $C^1_{0,1},\dots,C^n_{0,1}$ of $C_{0,1}$ and approximating $C_1\approx h(C_{0,1})$. 
Given the choice of $h$, $R_0=\rho_0(-X_1-R_1+C_1)$ is approximated by its empirical estimate based on simulated iid copies with respect to $\P$ of 
\begin{align*}
&v_{-1}(\beta_1^{\P}-1)C_{-1,1}+\sqrt{v_{-1}}\sigma_1^{\P}\varepsilon_{-1,2}+\sqrt{v_0}\sigma_1^{\P}c\\
&+\beta_1^{\P}\bigg(v_0\beta_0^{\P}+\sqrt{v_0}\sigma_1^{\P}\varepsilon_{0,1}\bigg)
-h\bigg(v_0\beta_0^{\P}+\sqrt{v_0}\sigma_1^{\P}\varepsilon_{0,1}\bigg)
\end{align*}
Similarly, $C_0$ is approximated by, for each $\theta$ in a dense subset of $\partial \Theta$, simulating iid copies with respect to $\Q_{\theta}$ of 
\begin{align*}
&v_{-1}(\beta_1-1)C_{-1,1}+\sqrt{v_{-1}}\sigma_1\varepsilon_{-1,2}^{\theta}+\sqrt{v_0}\sigma_1^{\P}c\\
&+\beta_1^{\P}\bigg(v_0\beta_0+\sqrt{v_0}\sigma_1\varepsilon_{0,1}^{\theta}\bigg)
-h\bigg(v_0\beta_0+\sqrt{v_0}\sigma_1\varepsilon_{0,1}^{\theta}\bigg),
\end{align*} 
estimating $\E^{\Q_{\theta}}[(R_0-X_1-R_1+C_1)^{+}]$ by the empirical mean, and computing the minimum of these expectations over the $\theta$ values. Finally, $V_0$ is estimated by the difference of the estimates of $R_0$ and $C_0$.   

Table \ref{tab:V0lowerupperbounds} shows numerical values for lower bounds $\underline{V}_0$ and for upper bounds $\overline{V}_0$ with the same parameter values as those considered for Case 1.  
Similarly to Case 1, the intervals $(\underline{V}_0,\overline{V}_0)$ are very narrow for $q$ small and therefore the upper bound $\overline{V}_0$ is an accurate estimate of $V_0$ when $q$ is small. 

\begin{table}[!ht]
\[\begin{array}{|l|ccc|}\hline
\multicolumn{4}{|c|}{\text{Case 1}} \\ \hline
& p=0.1 & p=0.5 & p=0.9 \\ \hline
q=0.10 & (1.452,1.491) & (1.562,1.624) & (1.686,1.787) \\
q=0.05 & (1.473,1.491) & (1.592,1.624) & (1.730,1.787) \\
q=0.01 & (1.490,1.491) & (1.618,1.624) & (1.772,1.787) \\
q=0.005 & (1.491,1.491) & (1.622,1.624) & (1.780,1.787) \\  \hline 
\multicolumn{4}{|c|}{\text{Case 2}} \\ \hline
& p=0.1 & p=0.5 & p=0.9 \\ \hline
q=0.10 & (1.470,1.513) & (1.595,1.666) & (1.734,1.856) \\
q=0.05 & (1.491,1.513) & (1.628,1.666) & (1.786,1.856) \\
q=0.01 & (1.509,1.513) & (1.656,1.666) & (1.835,1.856) \\
q=0.005 & (1.511,1.513) & (1.661,1.666) & (1.845,1.856) \\  \hline 
\end{array} \]
\caption{Case 1 and Case 2: lower and upper bounds $(\underline{V}_0,\overline{V}_0)$ rounded to three decimals, where the size of the parameter uncertainty region is determined by $r^2\leq F_{\chi^2(4)}^{-1}(p)$ and $\rho_t=\VaR_{t,q}$. Empirical estimates were based on iid samples of size $10^5$.}
\label{tab:V0lowerupperbounds}
\end{table}

\end{document}